\documentclass[ journal]{IEEEtran}
\IEEEoverridecommandlockouts
\usepackage{cite}
\usepackage{amsmath,amssymb,amsfonts, amsthm}
\usepackage{algorithmic}
\usepackage[ruled,vlined]{algorithm2e}
\usepackage{graphicx}
\usepackage{subfigure}
\usepackage{textcomp}
\usepackage{xcolor}
\usepackage{verbatim}
\usepackage{hyperref}
\hypersetup{hidelinks} 
\usepackage{setspace}
\usepackage[hyphenbreaks]{breakurl}
\def\BibTeX{{\rm B\kern-.05em{\sc i\kern-.025em b}\kern-.08em
    T\kern-.1667em\lower.7ex\hbox{E}\kern-.125emX}}

\newtheorem{theorem}{Theorem}

\newtheorem{corollary}{Corollary}
\usepackage{upgreek}

\usepackage{booktabs}
\newcommand{\RNum}[1]{\uppercase\expandafter{\romannumeral #1\relax}}
\newcommand{\tabincell}[2]{
	\begin{tabular}{@{}#1@{}}#2\end{tabular}
}

\begin{document}

\title{ CSQF-based Time-Sensitive Flow Scheduling in Long-distance Industrial IoT Networks\\
	%Cycle Tags Planning: Making CSQF Practical for Large-Scale Deterministic Networks\\
}

\begin{comment}
\author{
	\IEEEauthorblockN{Tao Huang\IEEEauthorrefmark{1}\IEEEauthorrefmark{2}, Yudong Huang\IEEEauthorrefmark{1}, Xinyuan Zhang\IEEEauthorrefmark{1}, Shuo Wang\IEEEauthorrefmark{1}\IEEEauthorrefmark{2}, Hongyang Du\IEEEauthorrefmark{3}, Hailong Zhu\IEEEauthorrefmark{1}\IEEEauthorrefmark{2}, Yunjie Liu\IEEEauthorrefmark{1}\IEEEauthorrefmark{2}}
		
	\IEEEauthorblockA{\IEEEauthorrefmark{1}State Key Laboratory of Networking and Switching Technology, BUPT, China}\\
	\IEEEauthorblockA{\IEEEauthorrefmark{2}Purple Mountain Laboratories, Nanjing, China}\\
	\IEEEauthorblockA{\IEEEauthorrefmark{3}School of Computer Science
		and Engineering, Nanyang Technological University, Singapore}
	%\IEEEauthorblockA{\IEEEauthorrefmark{3}Department of Systems and Computer Engineering, Carleton University, Ottawa}
	%\IEEEauthorblockA{Beijing Advanced Innovation Center for Future Internet Technology, Beijing, China}
}
\end{comment}

\author{
	\IEEEauthorblockN{Yudong~Huang,~Tao~Huang,~\IEEEmembership{Senior~Member,~IEEE},~Xinyuan~Zhang,~Shuo~Wang,\\Hongyang Du,~Dusit Niyato,~\IEEEmembership{Fellow,~IEEE},~and~Fei Richard Yu,~\IEEEmembership{Fellow,~IEEE} }

\thanks{ Y. Huang and X. Zhang are with the State Key Laboratory of Networking and Switching Technology, BUPT, Beijing, 100876, P.R. China  (e-mail: hyduni@bupt.edu.cn, zhangxinyuan0181@bupt.edu.cn).\protect\\
	
 	T. Huang and S. Wang are with the State Key Laboratory of Networking and Switching Technology, BUPT, Beijing, 100876, P.R. China. They are also with the Purple Mountain Laboratories, Nanjing, 211111, P.R. China (e-mail: htao@bupt.edu.cn, shuowang@bupt.edu.cn).\protect\\
 
	H. Du and D. Niyato are with the School of Computer Science
	and Engineering, Nanyang Technological University, Singapore (e-mail: hongyang001@e.ntu.edu.sg, dniyato@ntu.edu.sg).\protect\\
	
	F. Richard Yu is with the Department of Systems and Computer Engineering, Carleton University, Ottawa, ON K1S 5B6, Canada (e-mail:	richard.yu@carleton.ca).
}}

\begin{comment}
\author{\IEEEauthorblockN{1\textsuperscript{st} Given Name Surname}
\IEEEauthorblockA{\textit{dept. name of organization (of Aff.)} \\
\textit{name of organization (of Aff.)}\\
City, Country \\
email address or ORCID}
\and
\IEEEauthorblockN{2\textsuperscript{nd} Given Name Surname}
\IEEEauthorblockA{\textit{dept. name of organization (of Aff.)} \\
\textit{name of organization (of Aff.)}\\
City, Country \\
email address or ORCID}
\and
\IEEEauthorblockN{3\textsuperscript{rd} Given Name Surname}
\IEEEauthorblockA{\textit{dept. name of organization (of Aff.)} \\
\textit{name of organization (of Aff.)}\\
City, Country \\
email address or ORCID}
\and
\IEEEauthorblockN{4\textsuperscript{th} Given Name Surname}
\IEEEauthorblockA{\textit{dept. name of organization (of Aff.)} \\
\textit{name of organization (of Aff.)}\\
City, Country \\
email address or ORCID}
\and
\IEEEauthorblockN{5\textsuperscript{th} Given Name Surname}
\IEEEauthorblockA{\textit{dept. name of organization (of Aff.)} \\
\textit{name of organization (of Aff.)}\\
City, Country \\
email address or ORCID}
\and
\IEEEauthorblockN{6\textsuperscript{th} Given Name Surname}
\IEEEauthorblockA{\textit{dept. name of organization (of Aff.)} \\
\textit{name of organization (of Aff.)}\\
City, Country \\
email address or ORCID}
}
\end{comment}
\maketitle

\begin{abstract}
Booming time-critical services, such as automated manufacturing and remote operations, stipulate increasing demands for facilitating large-scale Industrial Internet of Things (IoT). Recently, a cycle specified queuing and forwarding (CSQF) scheme has been advocated to enhance the Ethernet. However,  CSQF only outlines a foundational equipment-level primitive, while how to attain network-wide flow scheduling is not yet determined. Prior endeavors primarily focus on the range of a local area, rendering them unsuitable for long-distance factory interconnection. This paper devises the cycle tags planning (CTP) mechanism, the first integer programming model for the CSQF, which makes the CSQF practical for efficient global flow scheduling. In the CTP model, the per-hop cycle alignment problem is solved by decoupling the long-distance link delay from cyclic queuing time. To avoid queue overflows, we discretize the underlying network resources into cycle-related queue resource blocks and detail the core constraints within multiple periods. Then, two heuristic algorithms named flow offset and cycle shift (FO-CS) and Tabu FO-CS are designed to calculate the flows' cycle tags and maximize the number of schedulable flows, respectively. Evaluation results show that FO-CS increases the number of scheduled flows by 31.2\%. The Tabu FO-CS algorithm can schedule 94.45\% of flows at the level of 2000 flows. 
\end{abstract}


\section{Introduction}

The Internet has been prone to congestion, crumble, and long-tail delays in the past few decades. As an emerging branch of the Internet, the Industrial Internet of Things (IoT) is leading a new wave of industrial revolution and network innovation\cite{iiot}\cite{10293177}. Different from traditional best-effort  transmission of bit information, industrial time-critical applications, such as remote operations\cite{remote_control}, smart grid\cite{smart_grid}, and digital twin\cite{digital_twin}, require task-oriented delivery with bounded delay and jitter. By leveraging the booming 5G and edge computing technologies,  several works\cite{URLLC}\cite{9788573} have studied the predictable low-latency communication paradigm within the range of a factory or a local area. However, the possibility of providing long-distance deterministic forwarding services for  large-scale cyber-physical systems remains to be explored. 

In the sensor-controller-actuator environment of Industrial IoT, many time-sensitive data need to be collected and transmitted to the edge node of the factory or cloud server for further analysis and process. Some of them are aperiodic traffic, such as event alerts and critical message updates, while others exhibit the features of periodic sending. For example, remote control of a slave robotic arm requires issuing kinematic packets every 10 ms to synchronize the Cartesian coordinates, acceleration, and angular velocity information. According to the ITU report\cite{ITUR}, machine-type devices could reach 97 billion in 2030. Rapidly growing IoT traffic leaves industrial proprietary buses facing bandwidth shortages,  prohibitive costs, and compatibility issues. The Ethernet is on the rise and promises to enable converged industrial networks.

%The emerging industrial revolution, i.e., Industry 4.0 , brings increased productivity, mass customization, reduced time-to-market, new business models, and leads the integration of information technologies (IT) and operational technologies (OT) \cite{URLLC}. With the proliferation of time-critical applications in large-scale cyber-physical systems, such as smart grid\cite{smart_grid}, remote operations\cite{remote_sur}, cloud PLC\cite{remote_control}, and digital twin\cite{digital_twin}, a deterministic forwarding service is highly desirable to transmit both real-time and best-effort traffic over one converged network, which arouses many challenges for the current Internet. The Internet is susceptible to congestion, collapse, and long-tail latency. Moreover, the Ethernet/IP is originally designed to provide best-effort services that cannot guarantee bounded delay, jitter, and zero packet loss for real-time applications.

%In response to the above challenges and trends, the deterministic networking has become a global research hotspot. 

On the one hand, the time-sensitive networking (TSN)  is a key candidate to enhance the Ethernet with real-time transmission capability. Since time synchronization struggles with guaranteeing accuracy in large-scale industrial sensor-controller-actuator environments, many TSN scheduling mechanisms, such as time-controlled gating, and cyclic queuing and forwarding (CQF), are limited to layer-2 local area networks (e.g., seven hops advised by the IEEE 802.1D protocol\cite{incre_tssdn}). Moreover, TSN generally assumes that the propagation delay of the link is virtually non-existent,  which is infeasible to be directly applied to long-distance networks.

%On the one hand, the IEEE time-sensitive networking (TSN) task group has put forward a set of standards, such as time synchronization, time-aware shaper (TAS), and cyclic queuing and forwarding (CQF),  to support real-time transmission in layer-2 local area networks. Since the network diameter in local area networks is limited to a small scale (e.g., seven hops recommended by the IEEE 802.1D standard\cite{incre_tssdn}), TSN assumes that all time on switches and end systems are precisely  synchronized and  the propagation delay over links is negligible. Due to these assumptions,  TSN is infeasible to be directly applied to large-scale networks. 

On the other hand, the IETF DetNet working group is currently extending the TSN mechanisms to layer-3 deterministic networks. One of the promising solutions is cycle specified queuing and forwarding (CSQF)\cite{load_balancing_csqf}. CSQF divides the sending time of an out-port into a series of equal time intervals; each time interval is called a cycle (noted as $T$). By tagging a sequence of cycle information to the packet headers, packets are expected to be transmitted at precise reserved durations along the path. Hence, once the cycle size and label stacks are determined, the end-to-end latency is bounded. More importantly, the end-to-end jitter is theoretically independent of hop count, which is controlled to no more than $2T$. 

%with relaxed prerequisites. CSQF  firstly relaxes the clock synchronization to the frequency synchronization; then it provides multiple receiving queues to sustain the long propagation delay. The packets are selected to enqueue at a specific receiving cycle tagged by the segment routing identifiers (SIDs), and the receiving queue is chosen to transmit cyclicly. Finally, the cycle mapping relationship of adjacent devices can be decided in advance; hence the end-to-end delay is bounded and predictable with the specific cycle information on each node along the path. 

CSQF is essentially an underlying primitive that needs the special hardware support\cite{PCSQ} for cycle mapping. CENI\cite{ceni} testbed evaluations have shed light on the performance of devices equipped with CSQF capabilities in real-world systems. However,  a cycle tag stack only explicitly defines the cyclic forwarding behavior of a single flow along a linear path. To make CSQF practical, network administrators necessitate an overarching strategy to effectively orchestrate network-wide traffic. Computing global scheduling cycle tags for time-sensitive flows continues to be a pressing challenge.

The scheduling task is analogous to the bin packing problem, with both being characterized as NP-hard\cite{itp}\cite{joint_large1}. While numerous TSN flow scheduling schemes\cite{9854862}\cite{explore_limits}  have been proposed with the time-controlled gating and CQF, they generally ignore the link delay and cannot reach the goal of long-distance deterministic transmission. 
%Furthermore, many of these solutions require precise time synchronization, making them inapplicable to CSQF. 
Recently, the authors of \cite{joint_large1} studied flow scheduling algorithms for CSQF,  but their focus was primarily on path generation, utilizing straightforward constraints that only addressed capacity and delay aspects, ignoring the cycle mapping and bounded queue length constraints. Moreover,  \cite{joint_large1} performs admission control at the flow level, which is incapable of adaptively adjusting the cycle tags, whereas we focus on the per-packet granularity and address the cycle tag computation problem.

It is not trivial to model CSQF. First, in the end-to-end latency calculation, long-distance link delays are hard to normalize and align with cyclic scheduling behaviors, which prevents hop-by-hop delays from being considered a bounded constant value. Second,  unlike traditional bandwidth allocation mechanisms, CSQF requires fine-grained time resource allocation. The industrial real-time applications send periodic traffic with different flow sending periods, which further arises challenges to unify cycle-related queue resource model. Third, since multiple flows will compete for the queue resources of each output port along the path, the cycle tag stack needs to be elaborately computed to avoid violating bounded queue length constraints. To our knowledge, this study is the inaugural effort to present the comprehensive integer programming mathematical model for the global scheduling of CSQF.

This paper targets to make the CSQF  practical for efficient network-wide flow scheduling in long-distance Industrial IoT networks.  The model hypotheses include: (1) Flow features and propagation delays are obtained prior, where flows are a set of periodic unicast flows marked by the highest priority. (2) CSQF-enabled switches adopt frequency synchronization, and the central controller applies zero phase deviation as a general scheduling basis. (3) All industrial terminals can send packets accurately, without considering the jitter from the source node. The main contributions are:

$\bullet$ We present a key insight that queuing delay could be bounded by binding the maximum queue length to a cycle time factor. Then, we establish a novel cycle tags planning (CTP) model\footnote{ The open-source code of the CTP model  is at https://github.com/Hyduni001/CTP-Model-for-CSQF.}, where the per-hop cycle alignment problem is solved by decoupling the long-distance link delay from cyclic queuing time. To avoid queue overflows under multi-flow aggregations, we discretize the network resources into cycle-related queue resource blocks and satisfy the bounded queue length requirement by constructing mapping constraints between flow and resource blocks.

% $\bullet$ We present a key insight that queuing delay could be bounded by corresponding the maximum queue length to a cycle time factor. Then, we establish a novel cycle tags planning (CTP) model \footnote{ The open-source code of the CTP model  is at https://github.com/Hyduni001/CTP-Model-for-CSQF.}, where the per-hop cycle alignment problem is solved by decoupling the long-distance link delay from cyclic queuing time. To avoid queue overflows under multi-flow aggregations, we discretize the network resources into cycle-related queue resource blocks and satisfy the bounded queue length requirement by constructing mapping constraints between flow and resource blocks.
  
  $\bullet$ We design a novel heuristic algorithm named flow offset and cycle shift (FO-CS)  to incrementally calculate the network-wide traffic's cycle tags and quickly generate per-flow scheduling solutions. By leveraging the Tabu search method, we also develop a Tabu FO-CS algorithm to maximize the  number of schedulable flows.
  
   $\bullet$ We evaluate the FO-CS algorithm under wide-area topologies in remote industrial control scenarios.  Simulation results demonstrate that FO-CS improves the scheduling flow number by 31.2\% compared with the naive algorithm. Furthermore, the Tabu FO-CS algorithm can schedule 94.45\% of flows at the flow level of 2000 under the realistic topology of Internet2. 

The structure of this paper is as follows: Section II provides a brief overview of the CSQF under Industrial IoT scenarios. Section III delves into the motivation. In Section IV, we present the formulation of the CTP model, while Section V details the design of the FO-CS algorithm. Section VI showcases our evaluation results. Related works are discussed in Section VII. The paper concludes in Section VIII.

\section{Background}

Currently, the Industrial IoT has brought a new era of network innovations. Unlike fast-fluctuation TCP applications, in the industrial automation and intelligent manufacturing scenarios,  real-time applications (e.g., remote operations and cloud control) issue promised aperiodic/periodic traffic and require the bounded latency of milliseconds and  jitter of microseconds\cite{URLLC}\cite{in_edge_control}.  The deterministic networking (DetNet) is generally considered as private WANs,  rather than replacing the Internet. The difference is that private WANs allow closed access control while  the Internet is managed by large groups of domains. DetNet is compatible with the existing Internet since the best-effort flow can still be transmitted by DetNet. 

Table \ref{table1} summarizes the examples of delay, jitter, data rate, and payload size requirements for time-sensitive tasks in typical DetNet scenarios, such as  discrete automation, process automation, and electricity distribution, according to 3GPP TS22261\cite{ts22261} and IETF RFC 8578\cite{use_case}. A small payload means it is less than or equal to 256 bytes, and a big payload generally does not exceed the MTU size. Note that all the values in this table are example values, which are varied in specific deployment configurations.

%Currently, the Industrial Internet has brought a new era of network innovations, and the concept of deterministic networking (DetNet) is actually used for networks under a single or a closed group of administrative control, such as factory interconnection and private WANs, but not for large groups of domains such as the Internet \cite{detnet_arc}. Unlike fast-fluctuation TCP applications, in the factory interconnection and flexible manufacturing scenarios,  the real-time applications (e.g., remote control and cloud PLC) send promised aperiodic/periodic traffic and require the bounded delay of milliseconds and  jitter of microseconds\cite{URLLC}\cite{in_edge_control}. Note that DetNet will not replace the Internet, and it is compatible with the existing Internet as the best-effort traffic can still be transmitted by DetNet.

\begin{figure}[]
	\centering
	\includegraphics[width=3.5in]{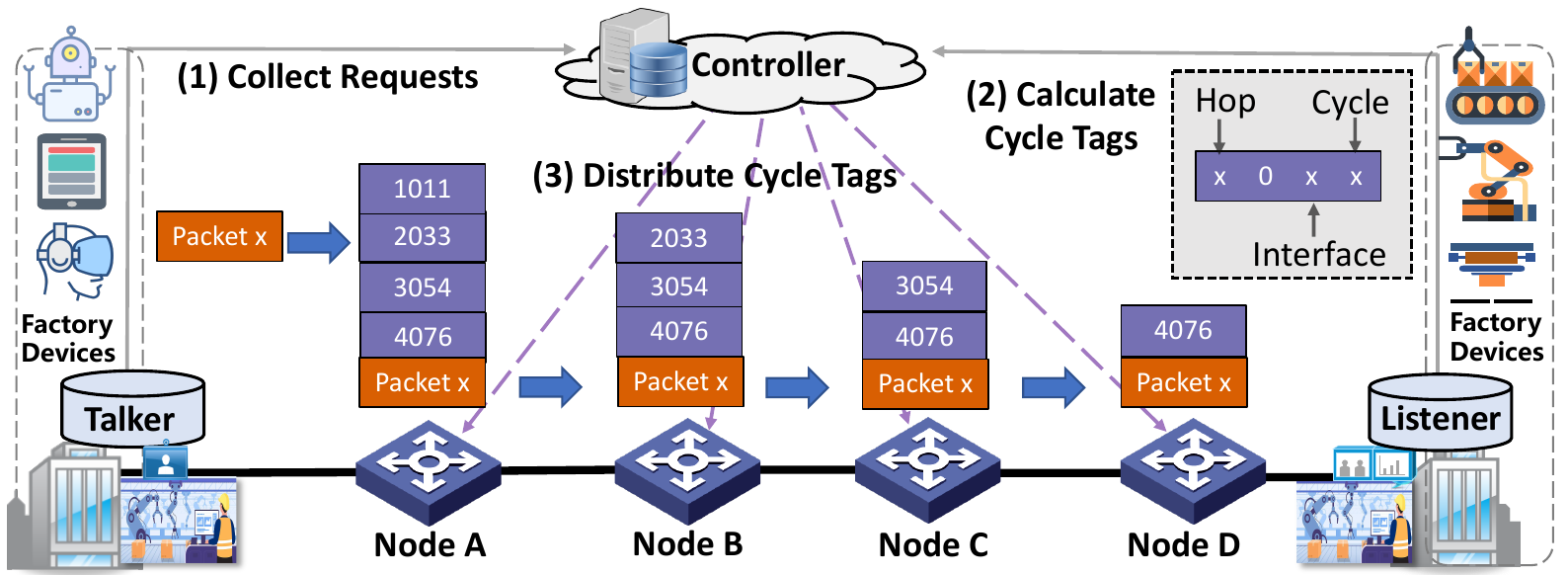}
	\caption{The working mechanism of CSQF in the factory interconnection or cloud PLC scenarios, where the control plane has to compute the cycle tags (SID labels) for every packet along the  path.}
	\label{fig:background}
\end{figure}

\begin{comment}
\begin{figure}[tp]
	\centering
	\includegraphics[width=3.5in]{./figures/obser_dn_ton}
	\caption{Different Ethernet-based methods to improve the quality of service (QoS) for Industrial IoT traffic.}
	\label{fig:obser_dn}
\end{figure}
\end{comment}

%DetNet is used for networks under a single or a closed group of administrative control, such as factory interconnection and private WANs, but not for large groups of domains such as the Internet\cite{detnet_arc}. DetNet will not replace the Internet, and it is compatible with the existing Internet as the best-effort traffic can still be transmitted by DetNet.

\begin{table}[]
	\centering
	\scriptsize
	\caption{Typical usecases and requirements for DetNet}
	\begin{tabular}{p{2.3cm}p{1.1cm}p{1cm}p{1.2cm}p{1.3cm}}
		\toprule[1pt]
		Scenarios / tasks & Latency & Jitter & Data rate & Payload size \\ \midrule[0.7pt]
		
		Discrete automation& 1-10 ms & 1-100 $\upmu$s & 1-10 Mbps &  Small to big\\
		
		Process automation-remote control & 50 ms & 20 ms & 1-100 Mbps & Small to big \\
		
		Process automation-monitoring& 50 ms & 20 ms &  1 Mbps & Small \\
		
		Electricity distribution-medium voltage & 40-100 ms& 1 ms &  10 Mbps & Small to big \\
		
		Electricity distribution-high voltage & 5-10 ms & 100 $\upmu$s&  10 Mbps & Small \\
		
		Electricity distribution-extra-high voltage & 5 ms & 10 $\upmu$s & / & Small \\
		
		Intelligent transport systems- backhaul& 10 ms & 20 ms & 10 Mbps & Small to big \\
		
		Tactile interaction & 5 ms & TBC & 10 Mbps & Small to big \\
		\bottomrule[1pt]
	\end{tabular}
	\label{table1}
\end{table}

\begin{figure*}[tp]
	\centering
	\subfigure[ The CQF mechanism. In each GCL, the 'o' denotes open and the 'c' denotes closed.]{
		\begin{minipage}[t]{0.3\textwidth}
			\centering
			\includegraphics[width=\textwidth]{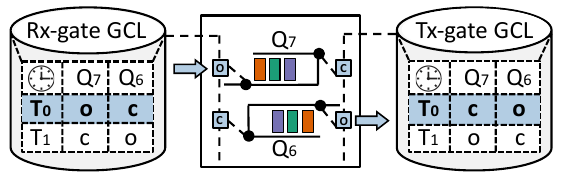}
			\label{fig2:CQF}
			% \caption{fig1}
		\end{minipage}
	}
	\subfigure[The end-to-end latency model for CQF. ]{
		\begin{minipage}[t]{0.3\textwidth}
			\centering
			\includegraphics[width=\textwidth]{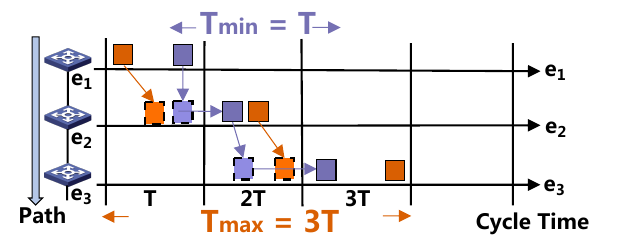}
			\label{fig2:CQF_cal}
		\end{minipage}
	}
	\subfigure[ CSQF instance. Three queues are rotated for transmission, each associated with a cycle.]{
		\begin{minipage}[t]{0.3\textwidth}
			\centering
			\includegraphics[width=\textwidth]{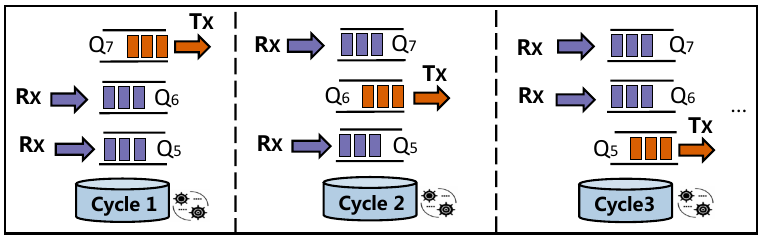}
			\label{fig2:CSQF}
		\end{minipage}
	}%
	
	\centering
	\caption{ Comparison of mechanisms evolving from local CQF to CSQF under long-distance Industrial IoT networks.
	}
\end{figure*}

With the efforts of the DetNet working group,  a cycle-based scheduling mechanism named CSQF is proposed. In CSQF, the key to achieve predictable latency and bounded jitter lies in the segment routing identifiers (SIDs)\footnote{The segment routing identifiers (SIDs) are equivalent to the cycle tags.}. Specifically, each SID is a label recording the packet's output port and sending cycle. For instance, in Fig. \ref{fig:background}, the SID of 4076 identifies that the packet x should be sent out from the cycle 6 and the port 7 of the hop 4 (node D). Several SIDs are combined into a label stack that indicates the dedicated forwarding path of the packet. Two data-plane instantiations of segment routing (SR) are SR over MPLS (SR-MPLS) and SR over IPv6 (SRv6). In practice, both SR-MPLS and SRv6 can be used for the SID of CSQF. The label of the sending cycle could be specified in the local segment layer, which is allocated from the Segment Routing Local Block (SRLB) according to RFC8402. 

The generation method and workflow of SIDs are as follows: (1) Once a talker at the source intends to issue time-sensitive traffic to a listener at the destination, it uploads the quality of service (QoS) requests to a centralized network controller. (2) The network devices learn from adjacent nodes to establish the cycle mapping tables. The controller calculates the feasible route and appropriate cycle parameters under the constraints of available spatial resources, temporal resources, and QoS requirements, and obtains the SID label stack corresponding to the path. (3) The control plane assigns the SIDs to the talker. Finally, CSQF-enabled network equipments reserve precise duration time for specified packets and forward them by parsing the topmost SID exhibited in the label stacks. 

%Proposed by the DetNet working group, the cycle specified queuing and forwarding (CSQF) is a cycle-based scheduling mechanism that attaches a list of  segment routing identifiers (SIDs)\footnote{The segment routing identifiers (SIDs) are equivalent to the cycle tags.} to a packet to offer bounded delay and lossless packet service delivery. As shown in Fig. \ref{fig:background}, when a talker wants to send time-sensitive flows to a listener, the workflow of connection setup is illustrated as follows: (1) A centralized network controller collects the requests of quality of service. (2) The controller generates the SIDs by calculating the feasible path and cycle parameters that satisfy the available cycle resources, bandwidth, and delay constraints. (3) The controller distributes the SID label stack to the talker and the network devices along the path. 

%The SID specifies the output interface and transmission cycle that a packet should be transmitted at each node (hop). For example, 4076 identifies cycle 6 of interface 7 at node 4. Thus, the CSQF-enabled devices can forward the time-sensitive packets at a precise reserved duration time by consuming the first SID available in the label stack of packet headers.  

Theoretically, assuming that each output port contains two cyclic queues and the flow is successfully scheduled, a simple but general calculation method for the maximum delay $D_{max}$ and the minimum delay $D_{min}$ is:

$$D_{max}= \sum_{i=1}^{h} (LD_{i}+PD_{i})+(h+1)T \text{,} \eqno{(1)}$$

$$D_{min}= \sum_{i=1}^{h} (LD_{i}+PD_{i})+(h-1)T \text{,} \eqno{(2)}$$

$$J_{e2e}= D_{max}-D_{min}=2T\text{,} \eqno{(3)}$$
where $LD$ is the link delay,  $PD$ is the process delay,  $h$ is the the number of hops, and $T$ is the cycle size. More importantly, since the packet can only fluctuate at the sending cycle of the first hop and the receiving cycle of the last hop, the end-to-end jitter $J_{e2e}$ is strictly limited to $2T$ regardless of network hops.

Furthermore, the segment routing (SR) used by CSQF is also a source routing technique that can implement the explicit route.  At the source node, all forwarding decisions towards the destination node have been carried in the SID labels. Hence, the intermediate and egress nodes do not need to maintain per-flow states. This characteristic enhances the scalability of CSQF, enabling it to schedule an extensive array of flows.

%Furthermore, segment routing is not only a viable approach to implement the explicit route, but also a source routing technology that does not need to maintain per-flow states at intermediate and egress nodes\cite{csqf}.  This feature helps to scale the CSQF to scheduling a large number of flows.

\section{Motivation}

This section presents the methods to construct the CSQF model and compute the global cycle tags.

\subsection{How to Model the CSQF?}

CSQF only specifies an underlying  packet scheduling mechanism on the data plane, which brings a significant problem of network-wide traffic planning to the control plane. To be more specific, how to construct the CSQF mathematical model and how to implement the SID calculation need to be solved. Next, We refer to the CQF in the range of local area to analyze the key points, i.e., link delays and cycle mapping relationships between adjacent nodes, in the process of modeling CSQF.

%CSQF only specifies an underlying device-level primitive. How to model the CSQF mechanism and how to achieve the network-wide cycle tag computation are significant challenges. The cyclic queuing and forwarding (CQF) model in the local area networks is a good reference for modeling the CSQF in the wide-area networks. Next, we compare and analyze the similarities and differences between the two mechanisms.

The features of local area include\cite{itp}: (1) The network diameter is restricted to a small scale, such as seven hops. (2) The time on switches  can be strictly synchronized. (3) The propagation latency can be ignored. Based on these features, CQF supposes that the cycle length exceeds the summation of transmission latency, link latency,  processing latency, and queueing latency.  As detailed in Fig. \ref{fig2:CQF},  the CQF-enabled switch has two queues, and each queue's behavior is controlled by the receive (Rx) and transmit (Tx) gates.  Determined by the gate control list (GCL), the two queues conduct enqueue and dequeue actions alternately. Thus, in a single cycle of $T$, the packet can be sent from the upstream device and received by the downstream device, then forwarded to the next-hop device in the subsequent cycle.	As depicted in Fig. \ref{fig2:CQF_cal}, if the packet is the last issued  and first arrived,  the minimum latency is $(H-1)\times T$, where $H$ is the number of hops. In contrast, a packet takes the maximum delay of $(H+1)\times T$ when it is the first sent and last received.

However, in wide-area networks, accurate time synchronization\cite{ieee8021as} is not easy to obtain and the varied propagation latency cannot be regarded as zero. Hence, it is not trivial to fix the per-hop latency to a  certain cycle time $T$. To address these challenges,  CSQF relaxes the synchronized time  to the frequency\cite{synce},  i.e., all switches maintain the same clock frequency but do not require the same initial phase, which introduces the initial phase deviation between devices. Then, CSQF enlarges the two queues of CQF to multiple queues, where one queue is for transmitting packets and the others for receiving packets. Fig. \ref{fig2:CSQF} depicts an instance of three-queues CSQF. Each queue corresponds to a cycle, and the queues dequeue packets in a periodic pattern. More than one receiving queues are utilized to absorb a certain amount of traffic bursts, but the mapping relationship between link delays and cyclic queuing behavior is still undefined.

%In a large-scale wide area network, the differences are that the network can not be strictly synchronized, and the link delay is unignorable. The one-hop delay is no longer a constant cycle time $T$ because it is not only decided by the queuing delay that depends on the maximum queue length but also the link delay changing in different links. To solve these problems, CSQF firstly relaxes the time synchronization to the frequency synchronization, while the initial phase deviation between nodes is also introduced. Then, it extends the CQF to multiple queues that one for transmitting and the remaining for receiving. An example of CSQF with three queues is shown in Fig. \ref{fig2:CSQF}, each queue corresponds to a cycle, and the transmitting queue is selected in a cyclic manner. Multiple receiving queues can be used to compensate for the delay caused by long-distance links, but the mapping relationship between link delays and cyclic queuing behavior is still undefined.

Based on the above analysis, we propose the cycle tags planning (CTP), an integer programming model for the CSQF. First, we argue that switches can be aware of the initial phase deviation and update it to the controller. The cycle mapping pattern still holds if we consider the phase deviation is 0 \cite{joint_large1}. Thus, for simplicity but without loss of generality, CTP assumes an initial phase deviation of 0 in the control plane as the basis of the scheduling model. Second, our CTP model decouples link delays from the cycle time $T$ to hold a constant queue rotation duration. Specifically, the irregular link delay is divided by the cycle time and rounded up. Third, to unify the link delay and cyclic queuing behavior, we model the underlying network resources as time-discrete queue resource blocks.  An auxiliary variable named arriving cycle is utilized to map the arriving time of packets to an appropriate receiving queue, which ensures aligned adjacent cycles. More details of CTP model are in Section  IV. 

 \begin{comment}
\begin{figure}[h]
	\centering
	\includegraphics[width=3.4in]{./figures/csqf.pdf}
	\caption{An example of CSQF with three queues corresponding to three cycles. }
	\label{fig:csqf_3q}
\end{figure}
 \end{comment}

\subsection{How to Calculate the Cycle Tags?}

Another key challenge for CSQF is how to compute the cycle tags for each flow under the CTP model. Since multiple flows will compete for the queue resources of each output port along the path, an efficient scheduling algorithm is needed  to not only maximize the number of flows that meet the deadline, but also avoid breaking the maximum queue length constraints. In addition, for different network scenarios, how to set appropriate configuration parameters (such as the queue number and the cycle size) also needs to be investigated.

%Another crucial problem for CSQF is how to compute the cycle tags, i.e., how to design the scheduling algorithm to solve the CTP model in large-scale deterministic networks. Particularly, when multiple flows compete for the queue resources of each outport port along the path, cycle tags are the key components to guarantee the precondition of bounded queue length on each switch without the gate and gate control list. Moreover, with the increasing flow number and topology size, scheduling algorithms' efficiency needs to be explored, e.g., how many queues and how long a queue are feasible and effective for specific scheduling scenarios.

\begin{figure}[t]
	\centering
	\includegraphics[width=3.5in]{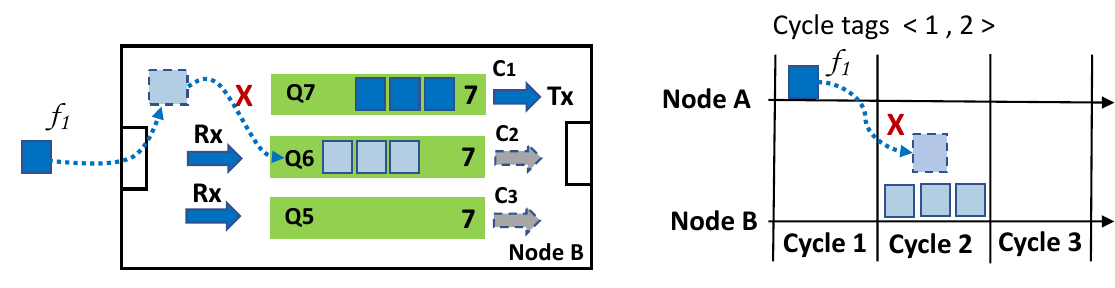}
	\caption{An access-aware scenario with the queue length of three packets. The node A is the access gateway device of the CSQF-enabled network domain. }
	\label{fig:guide_scenario}
\end{figure}

\begin{figure}[t]
	\centering
	\includegraphics[width=3.5in]{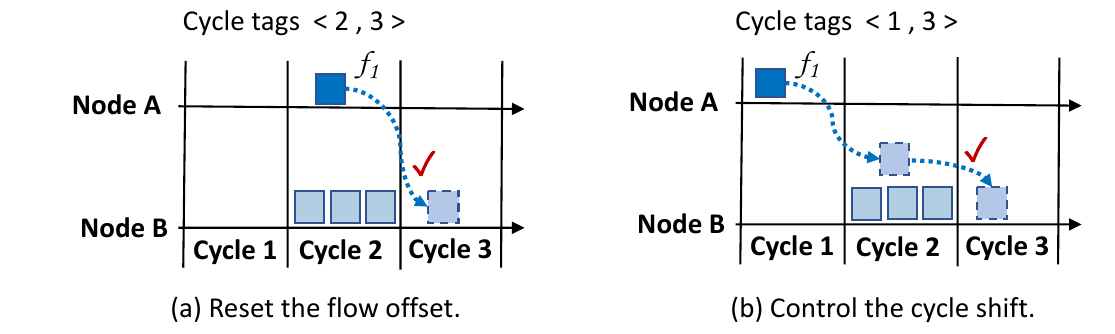}
	\caption{An example of the flow offset and cycle shift (FO-CS) method. The flow offset operation decides the packet's arrival time at the first hop. The cycle shift optimizes the selection of receiving queues at each node.  }
	\label{fig:algo_example}
\end{figure}

Specifically, there are two typical traffic access situations in real deployment scenarios: access-aware\cite{in_edge_control}\cite{towards}, and access-unaware. Access-aware means that access network devices are  gateways that can negotiate with terminal hosts and control when traffic will enter the CSQF-enabled network domain. Unawareness of the access indicates that the flow's arrival time of the first hop is uncontrollable. Thus, to make the cycle tags computation adapt to practical deployments, we propose two intuitive flow scheduling methods. One is to set the flow offset at the first hop (i.e., access-aware), and the other is to control the cycle shift at each receiving queue (i.e., access-unaware). 

To make it clear, we design a simple  access-aware scenario where the maximum queue capacity is limited to three packets. As shown in Fig. \ref{fig:guide_scenario}, Node A functions as the access gateway device, and there's a time-sensitive flow ($f_{1}$) directed towards the downstream node, Node B. It is not allowed to enqueue the receiving queue $Q_{6}$, because $Q_{6}$ is already at full capacity. One method, as demonstrated in Fig. \ref{fig:algo_example}(a), is to change the flow offset. By delaying the sending start time of flow $f_{1}$ at Node A by one cycle, the overflow at queue $Q_{6}$ can be circumvented. An alternative strategy involves adjusting the cycle shift, as seen in Fig. \ref{fig:algo_example}(b). Here, the flow $f_{1}$ is deferred to enqueue the queue $Q_{5}$. Consequently, both methods ensure that the flow $f_{1}$ is successfully scheduled and transmitted out during Cycle 3. The difference is that the flow offset operation can search over the entire sending period time series, while the cycle shift operation is limited to available queue numbers. Note that Node A can execute both flow offset and cycle shift operations if it is access-aware,  but can only execute the cycle shift operation if it is access-unaware. Moreover, if there is no available queue for choosing (e.g., all queues are full), the flow $f_{1}$ will fail to be scheduled. In this case, the flow needs to wait until some resources occupied by other flows are released, or the underlying network resources need to be expanded.

%Since the queue length is limited and corresponds to a cycle size, every time-sensitive flow must be elaborated scheduled in case any queue overflows. According to our early-stage investigations, two effective flow scheduling approaches are to set the flow offset at the first hop and control the cycle shift at each receiving queue. To stress it clearly, we build a simple scenario where a time-sensitive flow ($f_{1}$) is arriving at the downstream node (Node B) as depicted in Fig. \ref{fig:guide_scenario}. Each queue can hold three packets at most in a cycle time. Flow $f_{1}$ could have enqueued the receiving queue $Q_{6}$, but  $Q_{6}$ is already full. The first approach is resetting the flow offset as depicted in Fig. \ref{fig:algo_example}(a). When the sending time of flow $f_{1}$ at Node A is offset by one cycle, that is, starting from Cycle  2, the overflow of queue $Q_{6}$ can be avoided. The second approach is controlling the cycle shift at the receiving queue as depicted in Fig. \ref{fig:algo_example}(b),  i.e., making the  $f_{1}$ enqueue the other queue $Q_{5}$. Finally, the flow $f_{1}$ can be scheduled successfully and transmitted out at the time of Cycle 3.

%\subsection{Design Guideline}

Inspired by the aforementioned exploration, our objective is to devise a heuristic flow scheduling algorithm called flow offset and cycle shift (FO-CS) that possesses the following characteristics:

\textbf{Predictable performance}: The algorithms should adhere to the constraints of bounded queue length and ensure time-sensitive flows reach their destination within the stipulated deadline and jitter.

\textbf{High scalability}: Our approach ought to exhibit high scalability, scheduling thousands of time-sensitive flows across long-distance industrial IoT networks.

\textbf{Feasible execution time}: Employing heuristic methods, our scheme aims for a practical algorithm execution time.

 The FO-CS algorithm are elaborated in Section  V.

\section{Cycle Tags Planning}
In this section, we detail the CTP\cite{csqf_huang} model, including the integer programming inputs, decision variables, and core constraints.

\subsection{CTP Model}
We denote sets using calligraphic letters (e.g., $\mathcal{V, E, F}$) and vectors with bold letters. The physical network topology is conceptualized as a graph $\mathcal{G}= \left \{  \mathcal{V},  \mathcal{E}\right \} $, where $\mathcal{V}$ represents vertices, encompassing switches $\mathcal{S}$  and hosts  $\mathcal{H}$. $\mathcal{E}$ symbolizes tuples indicative of network links, defined as  $\mathcal{E}= \left\{ (v_{i},v_{j}) \mid v_{i},v_{j}\in \mathcal{V}, i\neq j\right\}$, signifying a link exists between $v_{i}$ and $v_{j}$. Additionally, every edge possesses a link delay attribute of $LD$.

A loop-free path $p$ comprised of edges $e_{k}$ is depicted in \eqref{eq1}. Here,  $e_{1}$ signifies the edge associated with the output port of the first switch (hop) of the path, while $e_{j}$ is associated with the last switch:

\setcounter{equation}{3}

\begin{equation}
\begin{aligned}
&e_{k}\in p, \quad k \in \left\lbrace1,2, ..., j\right\rbrace, \quad p=\left \langle e_{1},e_{2},...,e_{j}  \right \rangle.
\end{aligned}
\label{eq1}
\end{equation}

A periodic TS flow is characterized as a single unicast traffic from a source node to a destination node. The collection of periodic TS flows is denoted as  $\mathcal{F}=\left\lbrace f_i \right\rbrace , i=\left\lbrace 0,...,n-1\right\rbrace$. The flow feature of each flow $f_i$ is defined as

\begin{equation}
f_i = \left\lbrace s_i, d_i, SP_i, PK_i, D_i, p_i, \Omega_i\right\rbrace, \quad \forall f_i \in \mathcal{F}
\end{equation}

\noindent where $s_{i}$ is the source node, $d_{i}$ is the  destination node, $SP_i$ is the flow's sending period\footnote{The flow's sending period is an inherent attribute determined by the industrial terminal. If a flow starts sending at 1.2 ms and the sending period is 4 ms, it will send packets at 1.2 ms, 5.2 ms, 9.2 ms, and so on.}, $PK_i$ is the packet number, $D_i$ denotes the specific time by which the flow should reach its destination. Given their per-flow granularity, both the path information $p_i$ and the adjustable flow offset $\Omega_i$ are incorporated into the flow's attributes. The basic unit of $\Omega_i$ is the cycle time.

\begin{figure}[tp]
	\centering
	\includegraphics[width=3.5in]{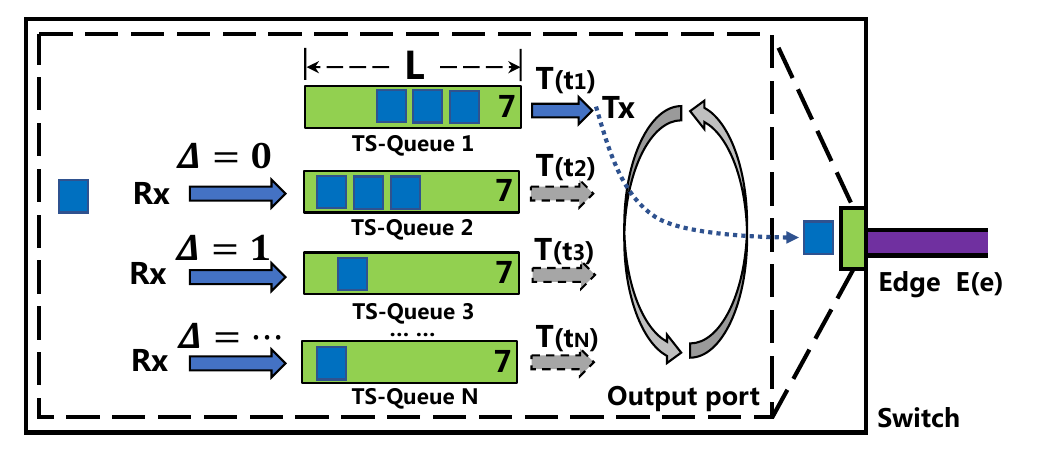}
	\caption{The CTP model from the perspective of the output port.}
	\label{fig:csqf_port}
\end{figure}

\begin{figure}[tp]
	\centering
	\includegraphics[width=3.5in]{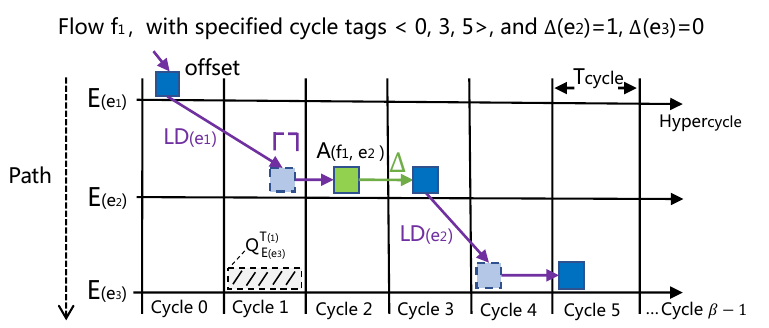}
	\caption{An illustration of CTP from a holistic flow perspective. The dashed gray block is an instance of queue resource blocks. For simplicity
		but without loss of generality, we draw frequency synchronization
		with the same initial phase (i.e., time synchronization). }
	\label{fig:csqf_flow}
\end{figure}

\begin{table}[]
	\caption{Graph related and flow related parameters}
	\centering
	\label{tab:graph related parameters}
	\setlength{\tabcolsep}{2mm}{
		\begin{tabular}{ll}
			\toprule
			$ \mathcal{E} \subset \mathbb{N}$&set of edges \\
			$ \mathcal{V} \subset \mathbb{N}$&set of vertexes\\
			%$\textbf{A}_{EE}\in\left\{ 0,1 \right\}^{\mid\mathcal{E}\mid \times\mid\mathcal{E}\mid}$&sparse edge-edge adjacency matrix\\
			$ \mathcal{F}\subset \mathbb{N}$&set of flow indexes \\
			$ s_{i}\in \mathcal{V}^{\mid F \mid}$& \tabincell{l}{source of flow $f_{i}$}\\
			$ {d}_{i}\in \mathcal{V}^{\mid F \mid}$& \tabincell{l}{destination of flow $f_{i}$}\\
			$ {SP}_{i}\in \mathbb{N}^{\mid F \mid}$& \tabincell{l}{sending period of flow $f_{i}$}\\
			$ PK_{i}\in \mathbb{N}^{\mid F \mid}$& \tabincell{l}{packet number of flow $f_{i}$}\\
			$ {D}_{i}\in \mathbb{N}^{\mid F \mid}$& \tabincell{l}{ deadline of flow $f_{i}$}\\
			${p}_{i}\in \mathcal{P}^{\mid F \mid}$& \tabincell{l}{ path of flow $f_{i}$}\\
			$ \Omega_i\in \mathbb{N}^{\mid F \mid}$& \tabincell{l}{ flow offset of flow $f_{i}$}\\
			$  \Delta_{(f_{i}, e_{k})}$& \tabincell{l}{ cycle shift of flow $f_{i}$ at edge $e_{k}$}\\
			$ LD_{(f_{i}, h)}$& \tabincell{l}{ flow $f_{i}$ traverses the $h$ hop of a path\\with link delay of $ LD_{(f_{i}, h)}$}\\
			$  Q_{E_{(e_k)}}^{T_{(t)}}$& \tabincell{l}{ resource block at edge $E_{(e_k)}$ and cycle $T_{(t)}$}\\
			$  L\in \mathbb{N}$& \tabincell{l}{ maximum queue length}\\
			$  N\in \mathbb{N}$& \tabincell{l}{ maximum queue number}\\
			$  T_{cycle}\in \mathbb{N}$& \tabincell{l}{ cycle size at the output port}\\
			$Hyper_{cycle} $& \tabincell{l}{hyper-cycle, which is equal to the least\\ common multiple of all flows' sending periods}\\
			$  \beta$& \tabincell{l}{ the number of  $T_{cycle}$ in a $Hyper_{cycle}$ }\\
			%$\lambda \in (0, 1]$& \tabincell{l}{The proportion of time slot resources\\ occupied by NWS}\\
			%$\textbf{u}\left [f  \right ]\left [e  \right ]$& \tabincell{l}{space vector of the relation between \\ flow and edge}\\
			%$\textbf{t}\left [f  \right ]\left [e  \right ]$& \tabincell{l}{time vector of the relation between \\ flow and edge}\\
			\bottomrule
	\end{tabular}}
\end{table} 

The transmission time of each output interface is segmented into a series of equal time intervals, each lasting a duration of $T_{cycle}$, commonly referred to as a cycle. Each cycle corresponds to a queue. Thus, each output port consists of a sequence of cycle-related queue resource blocks that can be defined as $Q_{E_{(e_k)}}^{T_{(t)}}$, and the superscript of $T_{(t)}$ is a variable that denotes the specific location of the cycle in the sequence. Specifically, $t \in \left \{0,1,  ... , \beta-1 \right \}$, and  $T_{(1)}$ denotes the  Cycle 1 as exhibited in Fig. \ref{fig:csqf_flow}.  The parameters related to the graph and flow, which are used in the CTP problem formulation, are presented in Table \ref{tab:graph related parameters}.

 As illustrated in Fig. \ref{fig:csqf_port}, the buffer allocated for TS traffic is partitioned into $N$ TS-queues. Of these, one queue is designated for transmission while the other $N-1$ queues are for receiving. The buffer is used to store the packet data, while the metadata of packets is stored in the queue. Each of the TS-queues is assigned the topmost priority level of seven. As per the CSQF mechanism, the queue for transmission is chosen cyclically.  In addition, the queue capacity $L$ is constrained and correlates with the size of  $T_{cycle}$. A flow must determine the appropriate receiving TS-queue for enqueueing, upon arrival at an interface (edge). The variable $\Delta_{(f_{i},e_{k})}$ represents the distance between the chosen receiving queue and the transmitting queue, measured by cycle shift operations at the associated edge $e_{k}$. Typically, the $\Delta$ is an integer value that ranges from $0$ to $N-2$.

Fig. \ref{fig:csqf_flow} presents an illustration of the CTP model from a comprehensive flow perspective. Flow $f_{1}$ initiates transmission at the first hop, specifically edge $e_{1}$, during the flow offset of Cycle 0. Following a link latency equivalent to nearly two cycles, $f_{1}$ arrives at the edge $e_{2}$  in the Cycle 2. Then, it enqueues the second receiving queue, thus the  value of $\Delta(f_{1}, e_{2})$ is one cycle and $f_{1}$ will be transmitted at the edge $e_{2}$ in the Cycle 3. After that, it is received at the edge $e_{3}$ with $\Delta(f_{1}, e_{3})$ of zero and transmitted in the Cycle 5. In its entirety, $f_{1}$ traverses through the path $\left \langle e_{1}, e_{2}, e_{3} \right \rangle $ adhering to the designated cycle tags $\left \langle0, 3, 5  \right \rangle $. 

%For the sake of simplicity, our description is based on the alignment of cycle start times for all nodes. 

%Actually, due to frequency synchronization, the start time of the cycle will be different for different nodes. The controller can measure this deviation in advance and convert cycle time to the node's local time. Thus, we consider all cycles are aligned in the model.

%The set of periodic TS flows is denoted as  $F$, and the total number of flows is $n$. The flow feature is defined as

\begin{comment}
\begin{equation}
\begin{aligned}
&\forall f_{i} \in F,i \in \left \{ 0,1,2, ..., n-1 \right \}\\
 &f_{i}=\left \{ s_{f},  d_{f},  period,  pkt_{num}, deadline,  path, offset\right \},  
\end{aligned}\label{eq2}
\end{equation}
\end{comment}

%where $s_{f}$ is the source node, $d_{f}$ is the  destination node, $period$ is the flow's sending interval, $pkt_{num}$ is the packet number, $deadline$ is the specific time before which the flow must arrive at the destination. The path information $path$ and the configurable offset $offset$ are appended to the flow feature since they are per-flow granularity. 

\subsection{Integer Programming Inputs and Variables}

Integer programming is a widely recognized mathematical method used to model network scheduling problems. The following are the inputs and decision variables.

1) Integer Programming Inputs: 

$\bullet$ Physical topology graph $\mathcal{G}=\left \{ \mathcal{V}, \mathcal{E} \right \}$.

$\bullet$ A collection of periodic time sensitive flows $\mathcal{F}$.

$\bullet$ Queue length $L$, queue number $N$, cycle size $T_{cycle}$.

2) Decision Variables:

$\bullet$ Flow offset, $\left \{ \Omega_i \mid  \Omega_i \in \Omega \right \}$.

$\bullet$ Cycle shift, $\left \{ \Delta_{(f_{i}, e_{k})} \mid f_{i} \in F, e_{k} \in p_i\right \}$.

\subsection{Core Constraints}

 \textit{1)  Hyper-cycle Constraint:}  Various flows have different sending periods. A hyper-cycle is used to determine the time series bounds for the scheduling solution. As shown in \eqref{eq3},  the hyper-cycle is defined as the least common multiple of all flows' sending periods ($\mathcal{SP}$). Thus, the sending pattern repeats from the global perspective, and we only need to make the flow satisfy the scheduling constraints in one hyper-cycle. The hyper-cycle is a widely-used term in real-time system fields\cite{itp}\cite{qbv_toc}. In the same vein, this paper focuses on scheduling periodic TS flows, where the sending period is considered as an integer within a reasonable range to comply with industrial tasks. The sporadic flows can be delivered with other QoS methods, such as asynchronous traffic shaping and frame preemption. Moreover, in specific industrial scenarios, such as 5G TSCAI\cite{9212141}, IEC 60802\cite{traffic_type_mapping}, and DetNet use cases\cite{use_case}, it is possible to register all types of devices and the periods of flows in advance. Consequently, the hyper-cycle can be established as the greatest value of the least common multiple derived from the flow periods of all known devices.

\begin{equation}
\begin{aligned}
Hyper_{cycle} = LCM(\mathcal{SP}). 
\end{aligned}
\label{eq3}
\end{equation}

 \textit{2)  Cycle Constraint:} In this study, we define the granularity of scheduling by the cycle time. This mandates that every sending period must be divisible by the cycle time. Consequently, the upper bound of  $T_{cycle}$ aligns with the greatest common divisor of the sending periods of all flows, as depicted in  \eqref{eq4}. On the other hand, the lower bound of $T_{cycle}$ must ensure that during its duration, all packets stored in a queue are dispatched. This duration encompasses the cumulative value of the processing latency $\delta_{p}$, maximum queuing latency $\delta_{q}$, and transmission latency $\delta_{m}$,  , as represented in \eqref{eq5}.

\begin{equation}
\begin{aligned}
&\max(T_{cycle})=GCD(\mathcal{SP}).&
\end{aligned}
\label{eq4}
\end{equation}

\begin{equation}
\begin{aligned}
\min(T_{cycle})=&\lceil \delta_{p} +\max \left (  \delta_{q} + \delta_{m}\right )\rceil  \\
=&\lceil  \delta_{p}  + \frac{L \times MTU}{Bw} \rceil . 
\end{aligned}
\label{eq5}
\end{equation} 

In (8), $Bw$ is the link capacity,  $MTU$ is the maximum transmission unit of packet size, and $L$ is the maximum queue length. Typically, the value of $T_{cycle}$ is chosen to be greater than its minimum, and $Hyper_{cycle}$ is divisible by $T_{cycle}$.  Then, $\beta$ is utilized to denote the total count of  $T_{cycle}$ in a $Hyper_{cycle}$ as presented in \eqref{eq6}.
\begin{equation}
\begin{aligned}
&\beta = \frac{Hyper_{cycle}}{T_{cycle}}.&
\end{aligned}
\label{eq6}
\end{equation}

 \textit{3) Offset Constraint:} The offset of each flow determines the sending initial cycle at the first switch device ($ \Omega_i$) linked to the terminal node, which must be less than the defined sending period of the flow, as outlined in \eqref{eq7}. In other words, the TS flow should be entirely forwarded  within the current sending period to avoid stream interference that may occur during adjacent sending periods.
\begin{equation}
\begin{aligned}
\forall i &\in \left \{ 0,1,2, ..., n-1 \right \},  \  \ 0\le \Omega_i < \frac{SP_{i}}{T_{cycle}} . 
\end{aligned}
\label{eq7}
\end{equation}

 \textit{4) Arriving Cycle Constraint:} As defined in \eqref{eq8}, $A_{(f_{i},e_{k})}$ denotes the arriving cycle when the flow is arriving at the output port of a node. $t_{A_{(f_{i},e_{k})}} $ is the specific time point that packets may arrive. This constraint is satisfied by default if the cycle size is set to no less than the minimal value of $T_{cycle}$ defined in \eqref{eq5}. 
 \begin{equation}
\begin{aligned}
\forall & i \in \left \{ 0,1,2, ..., n-1 \right \},\  \forall  e_{k} \in p_{i}, k \in \left \{ 1,2, ..., j \right \},\\
 &A_{(f_{i},e_{k})}=\Omega_i  + {\sum_{h = 0}^{k-1}}\left \lceil  \frac{ LD_{(f_{i}, h)}}{T_{cycle}} \right \rceil + { \sum_{h= 0}^{k-1}}\Delta_{(f_{i}, h)}, \\
 & A_{(f_{i},e_{k})} \times T_{cycle} \leq t_{A_{(f_{i},e_{k})}} \leq (A_{(f_{i},e_{k})} +1) \times T_{cycle}.
\end{aligned} 
\label{eq8}
\end{equation}

and $ LD_{(f_{i}, 0)} = 0$ ,  $ \Delta_{(f_{i}, 0)} = 0$.

 \textit{5) Deadline Constraint:} For each flow, all packets need to reach their destination within the designated deadline, i.e.,
\begin{equation}
\begin{aligned}
\forall  i &\in \left \{ 0,1,2, ..., n-1 \right \},\\
&1 + {\sum_{h=1}^{j}}\left \lceil  \frac{ LD_{(f_{i}, h)}}{T_{cycle}} \right \rceil + { \sum_{h=1}^{j}}\Delta_{(f_{i}, h)}\le \frac{D_{i}}{T_{cycle}}. 
\end{aligned}
\label{eq9}
\end{equation} 

 \textit{6) Flow and Edge-Cycle Mapping Constraint:} This constraint gives the mapping relationship between flow $f_{i}$ and queue resource block $Q_{E_{(e_{k})}}^{T_{(t)}}$ at the edge $e_{k}$ and cycle $t$. If a flow $f_{i}$ maps to a  $Q_{E_{(e_{k})}}^{T_{(t)}}$ successfully, i.e., the utilization of queue resource block does not exceed the queue length $L$, the value of the auxiliary variable $M\left ( f_{i}, Q_{E_{(e_{k})}}^{T_{(t)}} \right) $ is 1, as depicted in \eqref{eq10}. Otherwise, it is 0. 
 
\begin{equation}
\begin{aligned}
\forall & i \in \left \{ 0,1,2, ..., n-1 \right \},\\
\forall &e_{k}\in p_{i}, k \in \left \{ 1,2, ..., j \right \}, \forall t \in \left \{0,1,  ... , \beta-1 \right \},\\
&M\left ( f_{i}, Q_{E_{(e_{k})}}^{T_{(t)}} \right) = 1.\\
s.t.\\
&\lambda \in \left \{0,... ,  \frac{Hyper_{cycle}}{SP_{i}} - 1 \right \},  \\
&t =  \left (A_{(f_{i},e_{k})} + \Delta_{(f_{i}, e_{k})} + \frac{\lambda \times SP_{i}}{T_{cycle}}\right  ) \%  \beta. 
\end{aligned}
\label{eq10}
\end{equation} 

The actual transmitting cycle at the edge is determined by the sum of the arriving cycle $A_{(f_{i},e_{k})}$ and the cycle shift $ \Delta_{(f_{i}, e_{k})}$. Given that a flow can be issued from its source multiple times within a $Hyper_{cycle}$, the variable  $\lambda$ specifies the sending period to which the flow pertains. Moreover, to ensure that the arriving time of flows does not exceed a $Hyper_{cycle}$,  the transmitting cycle is subjected to a modulus operation with $\beta$.  In summary, the variable of $t$ encapsulates all the transmitting cycles a flow undergoes along its path within the range of a hyper-cycle.

%The arriving cycle $A_{(f_{i},e_{k})}$ plus cycle shift $ \Delta_{(f_{i}, e_{k})}$ is the value of actual transmitting cycle at the edge. Since a flow may be sent from the source for multiple times in a $Hyper_{cycle}$, the value $\lambda$ denotes which sending period the flow belongs to. Besides, the transmitting cycle is mod by the $\beta$ in case the arriving time of flows exceed a $Hyper_{cycle}$.  In summary, the value of $t$ indicates all the transmitting cycle for a flow along the path within the range of a hyper-cycle.

 \textit{7) Bounded Queue Length Constraint:} As defined in  \eqref{eq111}, if a flow $f_{i}$ is scheduled along the path successfully, i.e., for all the edges along the path the value $M\left ( f_{i}, Q_{E_{(e_{k})}}^{T_{(t)}} \right) $ is 1 while the deadline constraint and bounded queue length constraint are satisfied, then the value of $S_{i}$ is 1. We call this kind of flow a schedulable flow. Otherwise, it is 0.
 
  \begin{equation}
 \begin{aligned}
S_{i}=\left\{\begin{matrix}
1, & \text{if $f_{i}$ is a schedulable flow}  \\ 
0, & \text{otherwise}  
\end{matrix}\right.\text{.} 
 \end{aligned}
 \label{eq111}
 \end{equation}
 
The constraint on bounded queue length necessitates that the cumulative number of packets from all flows in any given queue not surpass the designated queue capacity, i.e.,
 \begin{equation}
\begin{aligned}
\forall & i \in \left \{ 0,1,2, ..., n-1 \right \},\\
\forall &e_{k}\in p_{i}, k \in \left \{ 1,2, ..., j \right \},\forall t \in \left \{0,1, ... ,  \beta-1 \right \},\\
&\sum_{i = 0}^{n-1}S_{i} \times M ( f_{i}, Q_{E_{(e_{k})}} ^{T_{(t)}} ) \times PK_{i} \leq L.
\end{aligned}
\label{eq11}
\end{equation}

In this paper, the optimization objective is to maximize the number of schedulable TS flows, which is described as:

%\begin{align}
%maximize \sum_{i = 0}^{n-1}S_{i}
%\end{align}

\begin{equation}
\mathbf{CTP:} \quad \max_{\mathbf{\Omega},  \mathbf{\Delta} } \ \sum_{i=0}^{n-1} S_i \qquad \quad
\end{equation}

\begin{equation*}
\mathbf{s.t.}  \quad(6),(7),(8), (10), (12),(13),(14),(15) 
\end{equation*}
where $\mathbf{\Omega}$ is the variable of flow offsets ( $\mathbf{\Omega} = \left[ \Omega_0, ..., \Omega_{n-1} \right]^T $), and $\mathbf{\Delta}$ is the variable of cycle shifts of flows ($\mathbf{\ \Delta} = \left[ \mathbf{\Delta}_0, ...,  \mathbf{\Delta}_{n-1} \right]^T $,  $\mathbf{\ \Delta}_{i} = \left[ \Delta_{i,1}, ...,  \Delta_{i,j} \right]^T$).

\subsection{Schedulability Analysis}
In this paper, schedulability refers to the number of time-sensitive flows that could be scheduled by the network under constrained network resources and desired quality of service. Many factors such as traffic characteristics, flow sequence, and topology size, may affect scheduling capabilities. One of the most important factors is the configurable parameter of queue length $L$. In (15),  it seems that as $L$ becomes larger, the number of packets each queue can hold increases, so the number of schedulable flows $\sum_{i = 0}^{n-1}S_{i}$ improves (i.e., $\sum_{i = 0}^{n-1}S_{i}\propto L$). Instead, we derive Theorem 1:

\begin{theorem}
When the total network resources are constant (due to limited on-chip memory), the schedulability increases as the queue length $L$ becomes smaller, i.e.,
 \begin{equation}
\begin{aligned}
\sum_{i = 0}^{n-1}S_{i}\propto \frac{1}{L}.
\end{aligned}
\label{eq13}
\end{equation}
\end{theorem}

\begin{proof}
	In (13), the  total network resources are the number of resource blocks $Q_{E_{(e_{k})}} ^{T_{(t)}} $ with the cycle $T_{(t)}$ and edge $E_{(e_{k})}$. For the flow mapping $M ( f_{i}, Q_{E_{(e_{k})}} ^{T_{(t)}} )$,  the search space is $\prod_{i=0}^{n-1}\frac{SP_{i}}{T_{cycle}}\times (N-1)^n$, where $n$ is the number of flows and $N$ is the queue number. Thus, we have:

 \begin{equation}
\begin{aligned}
\sum_{i = 0}^{n-1}S_{i}\propto \sum_{i = 0}^{n-1}M ( f_{i}, Q_{E_{(e_{k})}} ^{T_{(t)}} ) \propto \frac{1}{T_{cycle}}.
\end{aligned}
\label{eq14}
\end{equation}
In (8),  there is:

 \begin{equation}
\begin{aligned}
T_{cycle} \geq \delta_{p}  + \frac{L \times MTU}{Bw}  \propto L.
\end{aligned}
\label{eq15}
\end{equation} 
Hence, derived from (18) and (19),  there is $\sum_{i = 0}^{n-1}S_{i}\propto \frac{1}{L}$. Theorem 1 is proofed. 
\end{proof}

\begin{corollary}
	When the queue length $L$ is set to the minimum value of 1, (13)-(15) can be relaxed to the well-investigated frame isolation constraints \cite{qbv_toc}, in which the schedulability is maximized.
\end{corollary}
 
\begin{proof}
When $L$ takes the minimum value  of 1, it means that each cycle $T_{(t)}$ of each edge $E_{(e_k)}$ can only store one packet. We assume that the packet number of each flow (i.e., $PK_{i}$) is no more than one. Then, the bounded queue length constraints in (13)-(15) will be relaxed as follows:

$\forall f_{i}, f_{j}\in \mathcal{F},  \ e_{k}= (p_{i}\cap p_{j}) \in \mathcal{E}, i\neq j,$

$\forall x\in \left \{0, ...,  \frac{Hyper_{cycle}}{SP_{i}} - 1 \right\},$ 

$\forall   y\in \left \{ 0, ..., \frac{Hyper_{cycle}}{SP_{j}} - 1 \right\}:$ 

$$( A_{(f_{i},e_{k})}  +  \Delta_{(f_{i}, e_{k})} +  x \cdot  \frac{SP_{i}}{T_{cycle}} \geq A_{(f_{j},e_{k})}  +  \Delta_{(f_{j}, e_{k})} +  y \cdot  \frac{SP_{j}}{T_{cycle}} ) \vee$$

$$   ( A_{(f_{j},e_{k})}  +  \Delta_{(f_{j}, e_{k})} +  y \cdot  \frac{SP_{j}}{T_{cycle}} \geq A_{(f_{i},e_{k})}  +  \Delta_{(f_{i}, e_{k})} +  x \cdot  \frac{SP_{i}}{T_{cycle}} ) \eqno{(20)} 
$$

Equation (20) indicates that for any two flows $f_i$ and $f_j$ in any sending period, there are two ways to avoid scheduling conflicts. The first is that the transmission time slot of flow $f_i$ is after flow $f_j$, i.e., the left side of $\vee$. The second is that the transmission time slot of flow $f_i$ is before flow $f_j$, i.e., the right side of $\vee$.  Then, each flow independently occupies a cycle that can transimit an MTU-sized Ethernet frame. Since the packet size of flows (e.g., 100 Bytes) may be smaller than the MTU size, (20) still causes the waste of underlying time slot resources. Further, we assume the specific transmission duration of a flow $f_{i}$ is $\Theta _{i}$ and $t_{A_{(f_{i},e_{k})}} $ is the specific time point that packet may arrive. (20) is converted to:

$\forall f_{i}, f_{j}\in \mathcal{F},  \ e_{k}= (p_{i}\cap p_{j}) \in \mathcal{E}, i\neq j,$

$\forall x\in \left \{0, ...,  \frac{Hyper_{cycle}}{SP_{i}} - 1 \right\},$ 

$\forall   y\in \left \{ 0, ..., \frac{Hyper_{cycle}}{SP_{j}} - 1 \right\}:$ 

$$(t_{A_{(f_{i},e_{k})}}  +  x \cdot  SP_{i} \geq t_{A_{(f_{j},e_{k})}} +\Theta _{j}+  y \cdot  SP_{j} ) \vee$$

$$   ( t_{A_{(f_{j},e_{k})}}  +  y \cdot SP_{j} \geq t_{A_{(f_{i},e_{k})}} +\Theta _{i} +  x \cdot  SP_{i} ) \eqno{(21)} 
$$

Equation (21) is the frame isolation constraints \cite{qbv_toc}, in which the discrete equal cycles are converted to continuous time resources. Each frame occupies a minimum time slot resource that is the same as its transmission duration. Hence, the scheduling capacity is maximized.
\end{proof}

We adopt the bounded queue length constraints in (15) instead of frame isolation constraints in (21)  for the following reasons. First,  unlike LANs, it is hard to obtain fine-grained traffic characteristics (such as packet size or transmission duration) for all Industrial IoT flows in wide-area scheduling. Second, there is inevitable traffic aggregation in wide-area networks that causes queuing.  Finally, the execution time of the solution grows exponentially as the queue length gets smaller. Thus, the queue length $L$ should be set to an appropriate value to balance the schedulability and time costs.

\section{Algorithm Design}
This section elaborates on the implementation of cycle tags computation, including the path selection, the basic flow scheduling, and the enhanced flow sequence algorithm.

\subsection{Design Guidelines}
The most naive method to calculate the SID information for TS flows is formulating multiple constraints and inputting the flows' parameters to the integer programming solver. However, this method is memory-intensive and may cost several hours to obtain a premium scheme for long-distance Industrial IoT scenarios.  Fundamentally, the scheduling problem parallels the bin packing problem, which is NP-hard\cite{joint_large1}. Hence, obtaining the optimal solution within a practical timeframe might be unattainable.

To reduce the complexity and improve the schedulability, we opt for computation-intensive heuristic methods to address this issue. Similar to the idea in \cite{model_tabu}, we split the CTP problem into a basic flow scheduling problem and an enhanced flow sequencing problem. In incremental online scheduling scenarios, the basic flow scheduling algorithm targets to quickly compute cycle tags for an ordered set of flows. For offline optimization contexts, the enhanced flow sequencing algorithm, which leverages the meta-heuristic Tabu search, focuses on establishing a complete order of the flow set to optimize the count of schedulable flows.

%The basic flow scheduling algorithm based on a heuristic greedy approach finds a feasible solution for each flow with powerful scalability, and the enhanced flow sequence algorithm based on meta-heuristic Tabu search iterates the flow orders to approach the optimal solution.

\IncMargin{1em}
\begin{algorithm}[tp]
	\caption{FO-CS Algorithm}
	\label{algorithm:FOCS}
	\LinesNumbered
	\KwIn{Flow set: $\mathcal{F}$, Edge-cycle resource: $Q_{E_{(e_{k})}} ^{T(t)}$ }
	\KwOut{Scheduled flow set $\mathcal{F}_{suc}$}
	%$ F_{suc} \leftarrow Init\_result (\mathcal{F}, Q_{E_{(e_{k})}} ^{T(t)})$ according to (9)-(12), $\mathcal{F}_{fai}  \leftarrow \mathcal{F} \setminus \mathcal{F}_{suc}$ \;
	\For{ $f_{i}$ in $\mathcal{F}$ }{
		$init  \ \forall \ \Omega_{i} = 0,  \mathbf\Delta_{i}=0$\;
		$ Init\_result (f_{i}, Q_{E_{(e_{k})}} ^{T(t)})$ according to (12)-(15)\;
			\If{$S(f_{i})=1$}{
			$ F_{suc} \leftarrow f_{i} $\;
			Break\;
		}
		\Else{
			
		\While{$ \Omega_{i} <   \frac{SP_{i}}{T_{cycle}}$}{ 
			\For{ $e_{k}$ in $p_{i}$ and $M ( f_{i}, Q_{E_{(e_{k})}} ^{T(t)} )=0$ }{
				\While{$\Delta_{(f_{i}, e_{k})} \leqslant N-2 $ }{   
					$ \Delta_{(f_{i}, e_{k})}++$,
					$Compute(f_{i}, Q_{E_{(e_{k})}} ^{T(t)})$ according to (\ref{eq10})  \tcc*{ Cycle Shift Part}
					\If{$M ( f_{i}, Q_{E_{(e_{k})}} ^{T(t)} )=1$}{
						Break\;}
				}
			}
			\If{$S(f_{i})=1$}{
				$ F_{suc} \leftarrow f_{i} $\;
				Break\;
			}
			\Else{
				$ \Omega_{i} ++$   \tcc*{ Flow Offset Part}}
	
		$ F_{fai} \leftarrow f_{i} $\;
	}
}
	}
	
	return  $\Omega_{i},  \mathbf\Delta_{i},  \mathcal{F}_{suc}$\;
	% \While{not at end of this document}{
	% 　　if and else\;
	% 　　\eIf{condition}{
	% 　　　　1\;
	% 　　}{
	% 　　　　2\;
	% 　　}
	% }
	% \ForEach{condition}{
	% 　　\If{condition}{
	% 　　　　1\;
	% 　　}
	% }
\end{algorithm}
\DecMargin{1em}

\subsection{Path Selection}
Generally, periodic TS flows account for a relatively small portion of bandwidth (e.g., 7\% for scheduled flows in literature\cite{URLLC} and no more than 20\% for critical control streams in literature \cite{use_case}) compared to BE flows. For example, the differential protection traffic is just 64 kbps [4], but requires millisecond-level delay and microsecond-level jitter. Thus, to meet such stringent delay demands, the link delay weighs more than the bandwidth. In this study, we focus on CSQF scheduling and employ a weighted shortest path routing approach based on Dijkstra's algorithm. The  link delay is denoted as the weight of each edge, and the path with the smallest end-to-end weight is selected as the routing path. The time complexity is $O(\left | V \right | \log \left | V \right | + \left | E \right |)$, where $\left | V \right |$ is the number of vertexes and $\left | E \right |$ is the number of edges. Moreover, plenty of works have been done on the realization of routing for time-sensitive networks\cite{incre_tssdn}, such as Equal Cost Multi-Pathing (ECMP), Maximum Scheduled Traffic Load (MSTL), and segment routing. A load balancing based on CSQF is presented in \cite{load_balancing_csqf}, which can be applied to reduce the single point of congestion with additional complexity. Routing optimization with heuristic algorithms will be considered as future work.

\subsection{Basic Flow Scheduling Algorithm}

Based on the insights from Section III.C, our basic scheduling algorithm encompasses the cycle shift (CS), flow offset (FO), and a combined FO-CS algorithm. The FO-CS algorithm integrates elements from both the FO and CS. The FO algorithm operates under the premise that cycle shift operations in network output ports are zero, catering to situations where switches lack more than two queues. Simultaneously, the CS algorithm posits that TS traffic is dispatched with unpredictable initial times (offsets), addressing instances where the access arrival time point is uncontrolled.

 The FO-CS algorithm sets both the offset value and cycle shift value  to zero, sequencing the flows in a predetermined order, as outlined in Algorithm  \ref{algorithm:FOCS}. For each flow $f_{i}$, the initial step determines whether $S(f_{i})$ is 1 according to constraints (12)-(15). If $S(f_{i})$ is 1,  $f_{i}$ will be put into $F_{suc}$ (lines 1-6). Subsequently, if the flow is deemed unschedulable, the algorithm identifies the node where the mapping was unsuccessful. It then incrementally adjusts the cycle at that specific node, ensuring it stays within the constraints of the receiving queue's capacity (lines 9-13). Once a flow $f_{i}$ is scheduled successfully along a route,  it will be placed into the $F_{suc}$ set (lines 14-16). Otherwise, the value of flow offset is increased by one  and re-operates the cycle shift part considering the constraint of offset(lines 17-18). If the flow still cannot be scheduled after traversing all the flow offset and cycle shift operations, the flow will be placed into $F_{fai}$ (line 19). Finally, the algorithm outputs scheduling parameters of $\Omega_{i}$ and $ \mathbf\Delta_{i}$ for each scheduled flow $f_{i}$ in schedulable flow set $F_{suc}$. The CS part is preferentially executed because it has a smaller search space and lower time complexity than FO. Note that the worst-case time complexity of this FO-CS algorithm is $O( \sum_{i=0}^{n-1}\frac{SP_{i}}{T_{cycle}} \times N \times H)$, where $n$ is the number of  flows, $N$ is the queue number at each output port, and $H$ is the number of hops. The queue number $N$ and the hop number $H$ are bounded in a specific network. Therefore, the worst-case time complexity is $O( \sum_{i=0}^{n-1}\frac{SP_{i}}{T_{cycle}})$ for the FO-CS algorithm.

The FO-CS algorithm is flexible and efficient within the following three aspects. First, the maximum queue length and queue number can be adjusted with the increase of flow number and topology size. Second, the FO algorithm and CS algorithm can be utilized respectively under different scheduling scenarios, such as symmetric and asymmetric topologies. In symmetric topologies, traffic aggregation can be greatly alleviated by first-hop shaping of access switches; thus the flow offset algorithm works well independently and reduces the computational overhead. Third, the FO-CS algorithm quickly generates incremental flow scheduling solutions,  which supports online long-distance Industrial IoT traffic transmission.

\IncMargin{1em}
\begin{algorithm}[tp]
	\caption{Tabu FO-CS Algorithm}
	\label{algorithm: TapuFOCS}
	\LinesNumbered
	\KwIn{Flow set: $\mathcal{F}$, Edge-cycle resource: $Q_{E_{(e_{k})}} ^{T(t)}$ }
	\KwOut{Scheduled flow set $\mathcal{F}_{suc}$}
	$ curr\_result \leftarrow Init\_result (\mathcal{F}, Q_{E_{(e_{k})}} ^{T(t)}) $ \;
		$ best\_order \leftarrow Curr\_order( curr\_result) $ \;
		$ best\_result \leftarrow curr\_result $ \;
		\While {Termination criteria not satisifed} {
		$ curr\_order  \leftarrow Neighbor(curr\_order) $ \;
		$curr\_result  \leftarrow Gen\_best(curr\_order)$ \;
  \If{$curr\_result$ better than  $best\_result$}{
  	$  best\_order \leftarrow curr\_order $ \;
  	$best\_result  \leftarrow curr\_result$ \;
  }
}
	return $best\_order, best\_result, F\_{suc}$\;
	% \While{not at end of this document}{
	% 　　if and else\;
	% 　　\eIf{condition}{
	% 　　　　1\;
	% 　　}{
	% 　　　　2\;
	% 　　}
	% }
	% \ForEach{condition}{
	% 　　\If{condition}{
	% 　　　　1\;
	% 　　}
	% }
\end{algorithm}
\DecMargin{1em}

\subsection{Enhanced Flow Sequence Algorithm}

The enhanced flow sequence algorithm creates a varied scheduling order for the flow set $F$, to maximize the resulting number of schedulable time-sensitive flows computed by the FO-CS algorithm. The maximum number of the search space for the sequencing problem is $n!$, where $n$ is the number of flows. Since the brute force approach cannot be scaled to large problem sizes,  it is unacceptable in real-world applications. There are many mete-heuristic algorithms to approximate the optimal solution. Due to the advantage of recording
the previously searched solutions, we utilize the Tabu search\cite{model_tabu} with a simplified Tabu list as a guided exploration of the solution space. Wide-area deployment often is incapable of knowing all traffic prior. Thus, feasible online solutions are more appropriate than optimal offline solutions. We leave other mete-heuristic offline searching strategy with domain-specific knowledge (DSK) as future work.

As outlined in Algorithm \ref{algorithm: TapuFOCS}, the Tabu FO-CS algorithm generates the initial results using a random order, and defines it as the current result (lines 1-3). Then, the neighborhood of the current order is derived from the exchanging mode\cite{itp}. In this mode, a subset of mapped flows in $F_{suc}$ are randomly removed by cleaning the queue resource block associated with these flows. Following this, the flows in $F_{fai}$ are integrated into $F_{suc}$ until no enough network resources can be allocated (lines 5-6). Exchange operations are recorded in (and removed from) the Tabu list to avoid local optima.  If the current result is better than the previous best result, it will be marked as the best result and selected as the next solution for the next iteration. To reduce the superfluous search, the termination criteria are set based on two conditions: either the number of iterations surpasses the maximum threshold $K$, or there is no improvement observed in the best result over the past $P$ iterations (line 4). The efficacy of the Tabu FO-CS algorithm can be adjusted by varying the values of $K$ and $P$. While setting larger values for both parameters might lead to superior outcomes, it could significantly prolong the computation time.

\begin{figure}[t]
	\centering
	\includegraphics[width=3.5in]{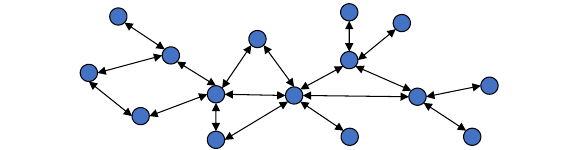}
	\caption{ER topology: an example of graph G(15, 18).}
	\label{fig:topology3}
\end{figure}

\begin{figure*}[]
	\centering
	\subfigure[ Naive \textit{vs} Flow offset and cycle shift (FO-CS).]{
		\begin{minipage}[t]{0.3\textwidth}
			\centering
			\includegraphics[width=\textwidth]{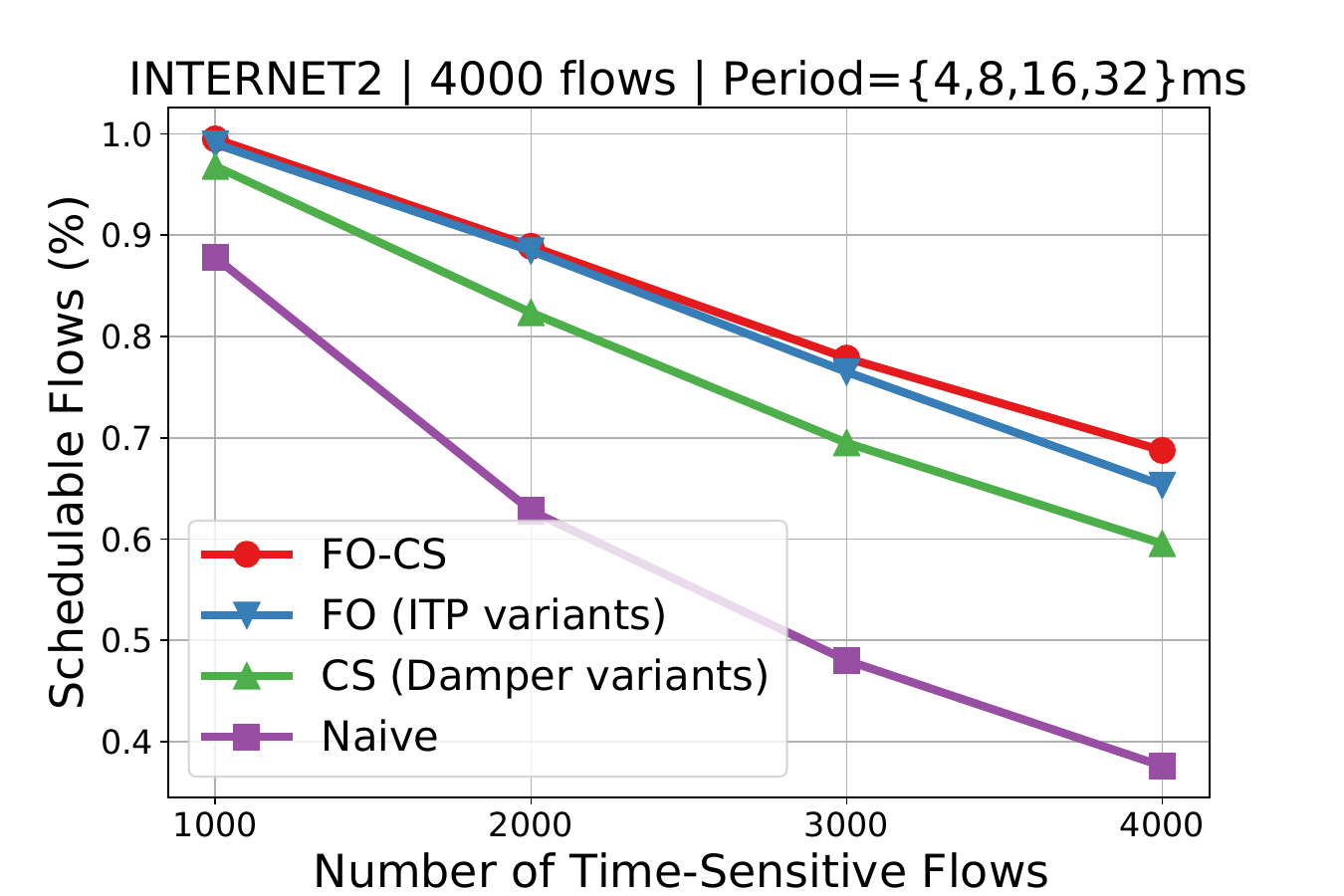}
			\label{fig10:evalutiona}
			% \caption{fig1}
		\end{minipage}
	}
	\subfigure[ Effect of queue number on cycle shift (CS).]{
		\begin{minipage}[t]{0.3\textwidth}
			\centering
			\includegraphics[width=\textwidth]{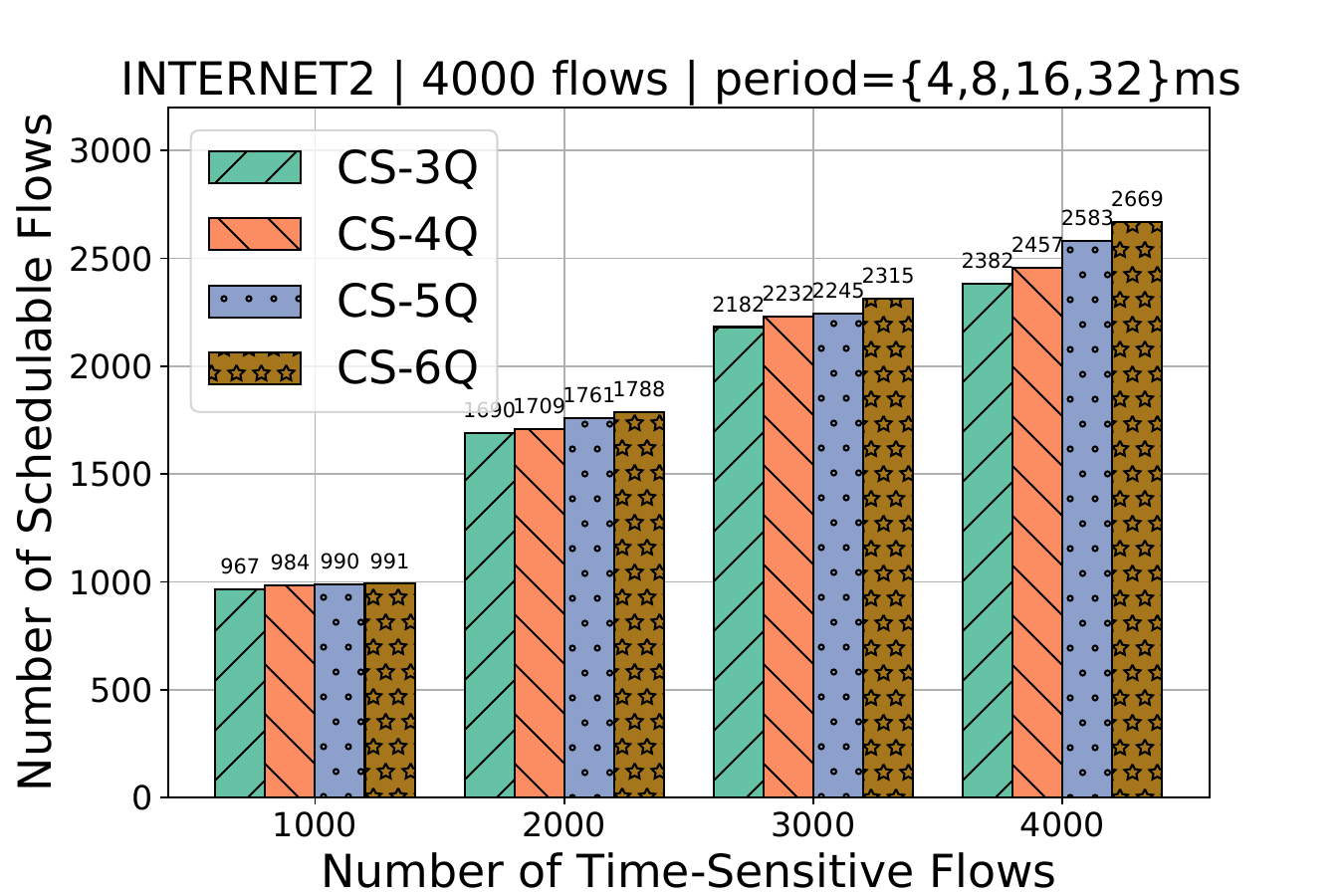}
			\label{fig10:evalutionb}
			% \caption{fig1}
		\end{minipage}
	}
	\subfigure[Effect of memory allocation on FO-CS.]{
		\begin{minipage}[t]{0.3\textwidth}
			\centering
			\includegraphics[width=\textwidth]{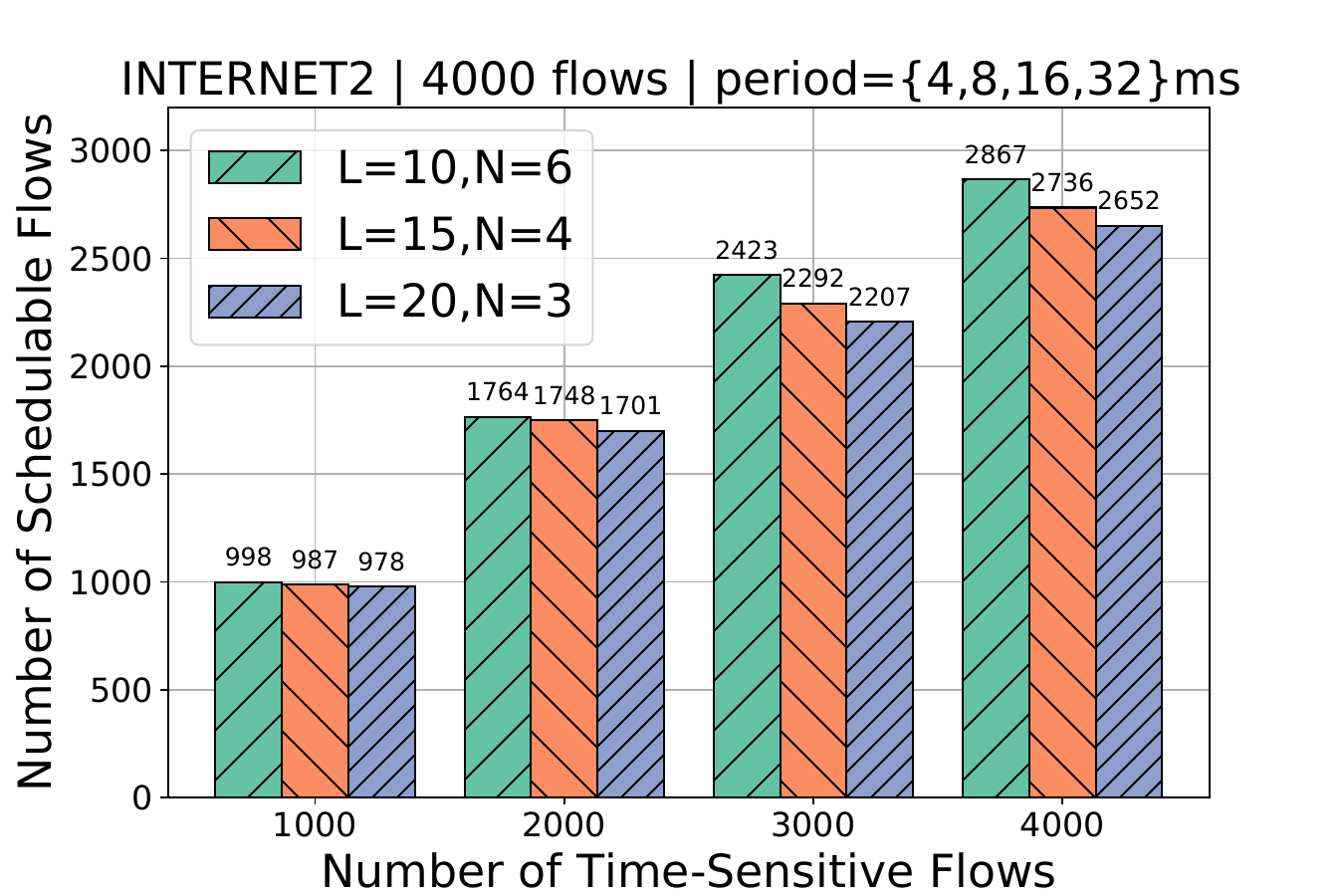}
			\label{fig10:evalutionc}
		\end{minipage}
	}
	\centering
	\caption{ Simulation results for the influence of FO-CS algorithms on  the number of schedulable flows. 
		%(a) The number of scheduled flows and execution time per flow vs No. of links in the topology. (b) The number of average scheduled flows with two different periods set in small topology. (c) The average execution time and the maximum execution time with two different periods set in small topology. 
	}
\end{figure*}

\section{Evaluations}
This section presents the simulation results of FO-CS and Tabu FO-CS algorithms. Firstly, we outline the simulation setup. Then, we analyze the algorithm performance, including the number of schedulable flows,  time costs, and scalability.

\subsection{Simulation Setup}

All evaluations are conducted on a computing device powered by four Intel i5-8259U CPUs, operating at a base clock speed of 2.3 GHz, and complemented with 16 GB of RAM.

\subsubsection{ Topology Selection} The INTERNET2\cite{internet2} stands as a pivotal advanced network where research and trials of emerging technologies are conducted. These innovations are then integrated into the next iteration of the Internet, including technologies like IPv6, MPLS, and QoS. Industrial control from remote locations frequently encompasses long-distance transmissions across various regions. To assess the efficacy of our proposed mechanism, we select a segment from the advanced Layer 3 service of INTERNET2's backbone in 2016 as a test topology\cite{internet2}. This segment comprises eight nodes situated in cities including Atlanta, Chicago, Cleveland, Ashburn, Washington, and New York. The geographical distance between these nodes serves as a proxy for the link length. The link delay is computed by dividing the link length by two-thirds the speed of light. To verify the effect of the link connectivity, we leverage the Erd$\ddot{o}$s–R$\acute{e}$nyi  (ER)  model to build a directed weighted random graph with fifteen vertices.  An instance of the $G_{(n, m)}$ random topology is depicted in Fig. \ref{fig:topology3}, where $n$ indicates the number of nodes and $m$ indicates the number of edges\cite{incre_tssdn}\cite{explore_limits}. The link delay is randomly selected from 0.1 ms to 2 ms.

\begin{comment}
\begin{figure}[t]
	\centering
	\includegraphics[width=3.5in]{./figures/internet2_part1}
	\caption{Part of 2016 Advanced layer 3 service of INTERNET2's backbone.}
	\label{fig:internet2}
\end{figure}
\end{comment}

\subsubsection{ Resource Configuration} In our experiments, we configure resources based on parameters such as link bandwidth, queue length, $T_{cycle}$ size, and queue count. Specifically, we allocate a reserved link bandwidth of 1 Gbps, establish each CSQF queue's length at 10, and define the $T_{cycle}$ size as 125 $\mu s$. If the time-sensitive packets do not require a relatively low deadline and the on-chip memory is surplus, the queue length can be set larger. The $T_{cycle}$ size is set to 125 $\mu s$ for three reasons.  First, the transmission delay for 10 MTU-sized packets is about 120 $\mu s$ according to the (8). Second, the minimum value of $T_{cycle}$ is selected to increase the search space of the cycle time. Third, the $T_{cycle}$  better be a divisor of $Hyper_{cycle}$ to simplify the calculation. For TS traffic, the queue count $N$ is set to vary between 2 and 6. When $N$ equals 2, it implies that additional cycle shifts cannot be accommodated.

\subsubsection{ Traffic Characteristics} For each flow $f$, its source and destination are determined by randomly selecting two distinct nodes with a uniform probability. Flows are created in alignment with the traffic patterns observed in wide area control and monitoring systems (as outlined in the IEC 61850) and the industrial machine-to-machine context of the DetNet use case\cite{traffic_type_mapping}\cite{use_case}\cite{traffictype}. Typically, the sending period of every flow is in the range of milliseconds, we emulate it by stochastically selecting one item from the assembly of $\left \{4, 8, 16, 32\right \}$ ms. The number of packets in each period for a given flow is chosen from the assembly of $\left \{1,2,3\right \}$. The deadline of flows ranges from 30 ms to 50 ms.

\begin{figure*}[]
	\centering
	\subfigure[ Naive \textit{vs} Flow offset and cycle shift (FO-CS).]{
		\begin{minipage}[t]{0.3\textwidth}
			\centering
			\includegraphics[width=\textwidth]{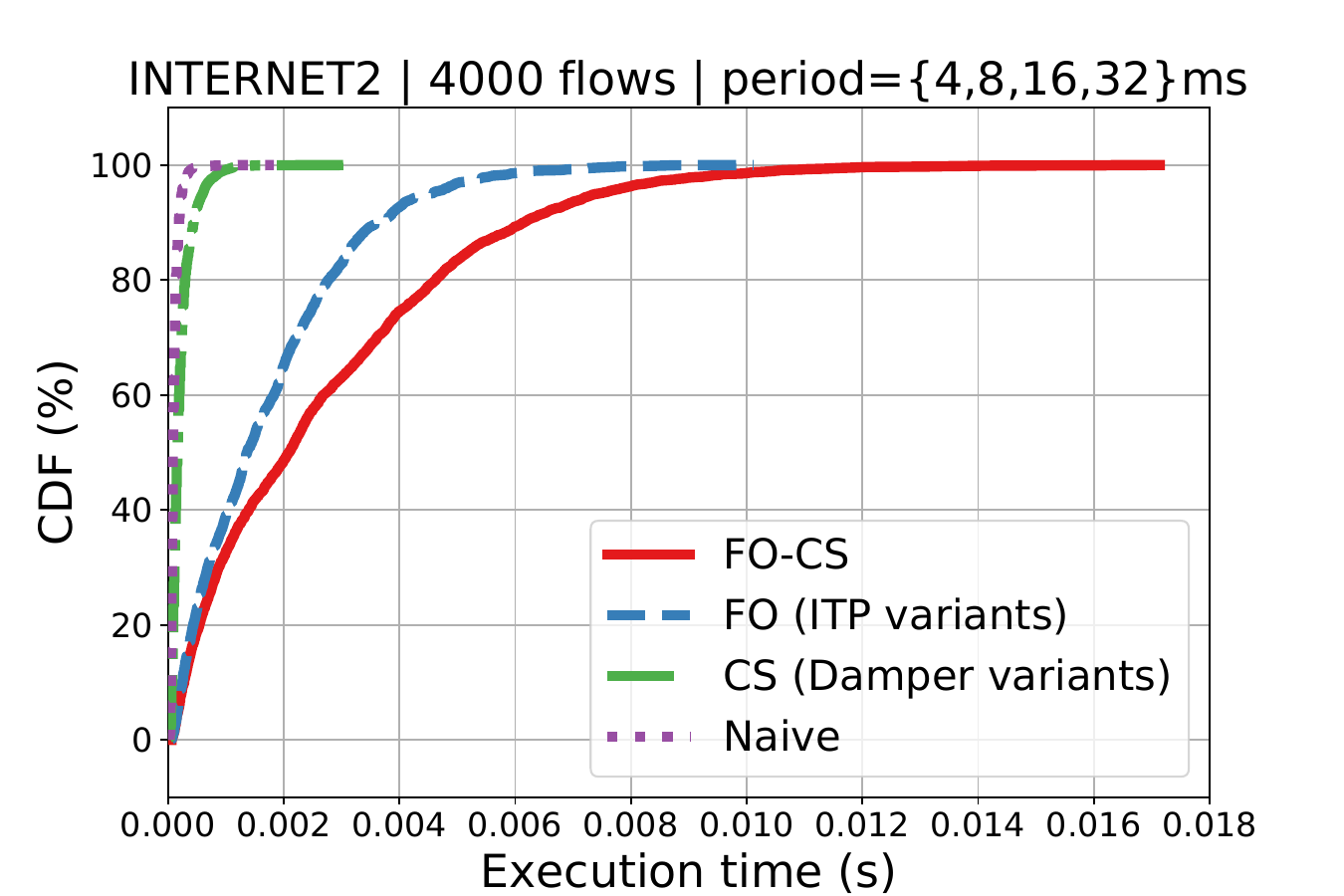}
			\label{fig11:evalutiona}
			% \caption{fig1}
		\end{minipage}
	}
	\subfigure[ Effect of queue number on cycle shift (CS).]{
		\begin{minipage}[t]{0.3\textwidth}
			\centering
			\includegraphics[width=\textwidth]{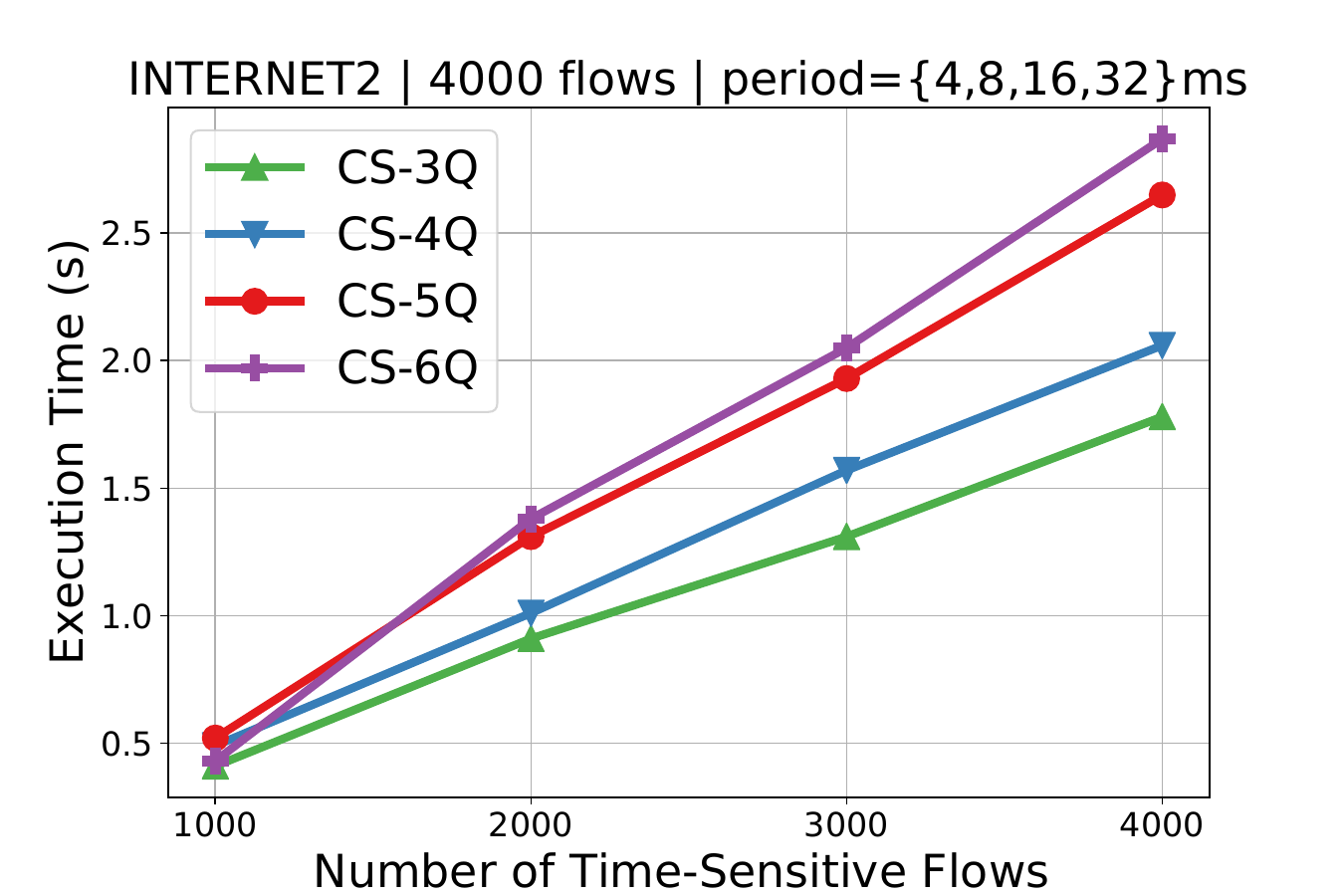}
			\label{fig11:evalutionb}
			% \caption{fig1}
		\end{minipage}
	}
	\subfigure[Effect of memory allocation on FO-CS.]{
		\begin{minipage}[t]{0.3\textwidth}
			\centering
			\includegraphics[width=\textwidth]{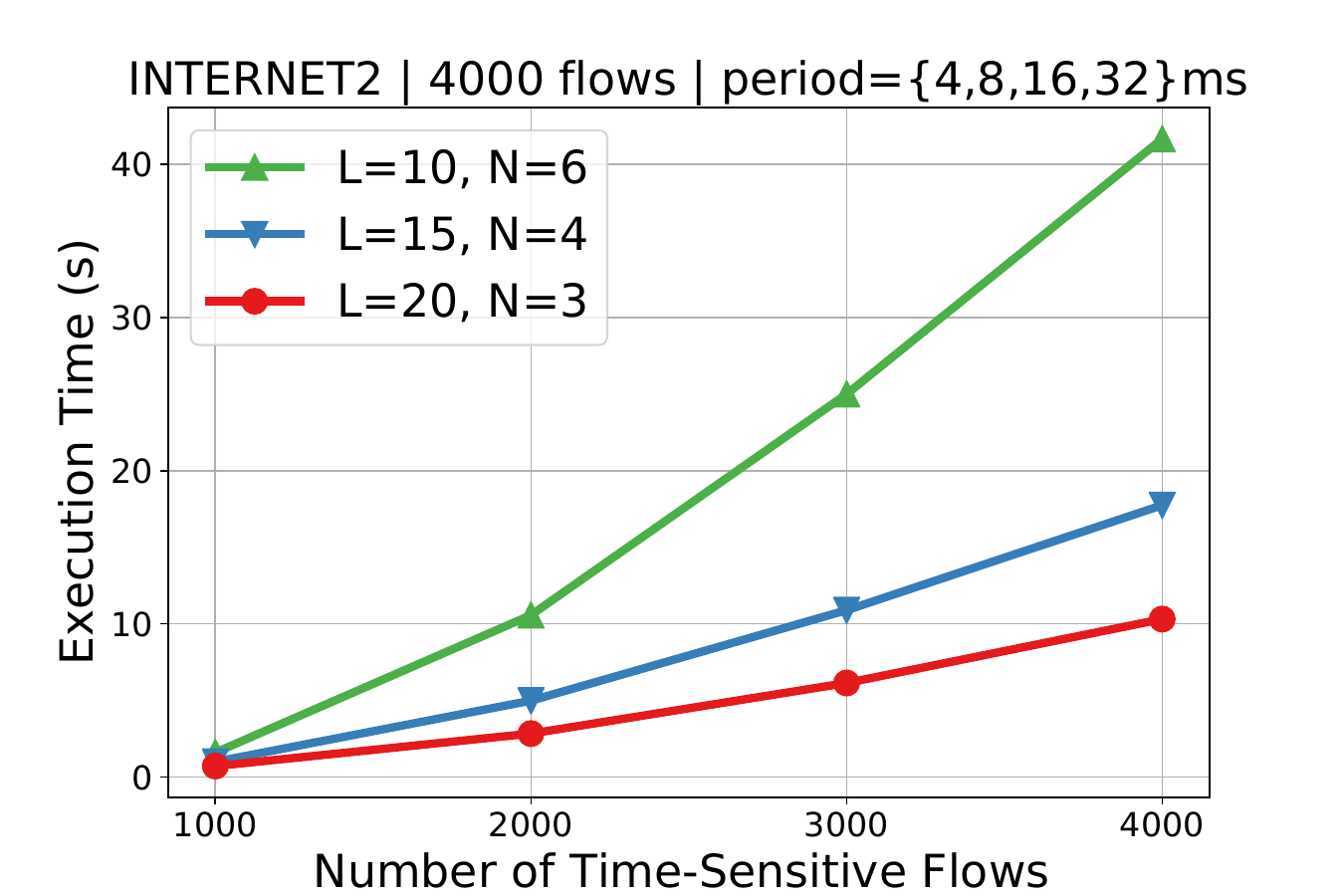}
			\label{fig11:evalutionc}
		\end{minipage}
	}
	\centering
	\caption{ Simulation results for the influence of FO-CS algorithms on the execution time. 
		%(a) The number of scheduled flows and execution time per flow vs No. of links in the topology. (b) The number of average scheduled flows with two different periods set in small topology. (c) The average execution time and the maximum execution time with two different periods set in small topology. 
	}
\end{figure*}

\subsection{Simulation Results for FO-CS}

\subsubsection{ Impact on the Schedulability} 
The schedulability is indicated by the number of schedulable flows. Under the condition of Internet2 topology, we partition the  flow counts into four groups from 1000 to 4000, and each group is executed five times to count average scheduled flows. The default queue number is set to 3.

Firstly, we juxtapose the performance of FO-CS algorithm with that of FO, CS, and naive algorithms, considering varying flow counts as presented in Fig. \ref{fig10:evalutiona}. Specifically, the naive algorithm dispatches packets as they're produced, without regulating the cycle shift or sending time. To compare the performances of CSQF with other  mechanisms,  the FO algorithm with two queues, which decides the packet's arrival cycle offset at the first hop, can be considered a benchmark of the ITP\cite{itp}. Damper\cite{9723445} defers packets for a specified amount of time hop by hop. The CS algorithm with more than two queues, which optimizes the selection of receiving queues at each node, can be considered a variant of the Damper. As we scale up the flow number, the advantages of FO-CS become more pronounced compared to the naive approach. The maximum improvement observed is a commendable 31.2\% over the naive approach and 9.2\% over the CS under dealing with 4000 flows. Fig. \ref{fig10:evalutiona} further indicates that FO's scheduling capability is nearly on par with FO-CS. This is because the flow offset operation offers a broader search space than the cycle shift, making offset control more impactful than managing the cycle shift at every hop. For individual flow scheduling in high-volume traffic scenarios, FO-CS remains valuable, offering the benefit of incremental cycle adjustments at each hop.

In situations where the flow offset is unmanageable at the access node, the CS algorithm can operate independently with varying queue numbers. Fig. \ref{fig10:evalutionb} examines the distinctions among the CS-3Q, CS-4Q, CS-5Q, and CS-6Q algorithms. The CS-6Q algorithm consistently outperforms the others across all flow levels. The performance disparity between CS-6Q and CS-3Q fluctuates between 2.4\% and 7.18\%. Notably, as flow numbers rise, the effectiveness of the CS algorithm improves more markedly with a large queue number.

Since the on-chip memory is limited,  Fig. \ref{fig10:evalutionc} evaluates the effect of memory allocation for FO-CS algorithm. Assuming the buffer is sixty packets at each output port, we divide the buffer into three solutions: L=10 and N=6, L=15 and N=4, L=20 and N=3. Evaluation results show that the queue length of 10 has better schedulability than the other two.  The maximal improvement is 7.2\% compared to the queue length of 20 when the flow number is 3000. Thus, fine-grained queue division can improve the scheduling capability.

\subsubsection{ Impact on the Execution Time} 
The algorithm execution time is independent of the flows' sending period. Before sending packets, the talker sends requests to the controller, ensuring that time slot resources are strictly reserved in all sending periods. Thus, the algorithm just needs to be executed once when a flow is initially considered to be accommodated. Fig. \ref{fig11:evalutiona} depicts the cumulative distribution function (CDF) of per-flow execution time for FO-CS, FO, CS, and naive algorithms.  For the FO-CS approach, each flow's execution time does not exceed 18 milliseconds, 90\% of per-flow execution times are under 6 milliseconds, and the total time for 4000 flows is about 20 seconds, revealing that all our algorithms have a feasible time for execution. The CS curve is on top of the FO and the FO-CS,  suggesting that managing the per-flow offset operation is more time-consuming than adjusting the per-node cycle shift.  Fig. \ref{fig11:evalutionb} shows the total execution time of the CS algorithm when scheduling 1000, 2000, 3000, and 4000 flows respectively. The total execution time of CS algorithms with different queue numbers varies from 0.41 seconds to 2.87 seconds. The cycle shift with a larger queue number will slightly increase the running time.  

Fig. \ref{fig11:evalutionc} shows the total running time of the FO-CS algorithm with different flow number groups and memory allocation solutions. When scheduling 1000 flows, the execution time for all solutions maintains at the same low level. However, the execution time of the first solution (L=10 and N=6) increases sharply when the flow number increases. At the flow level of 4000, the execution time of the first solution is twice that of the second solution and four times that of the third solution. Consequently,  the queue division is not as fine as possible when considering the execution time overhead.

\begin{figure*}[]
	\centering
	\subfigure[ Impact on the schedulable flow number.]{
		\begin{minipage}[t]{0.3\textwidth}
			\centering
			\includegraphics[width=\textwidth]{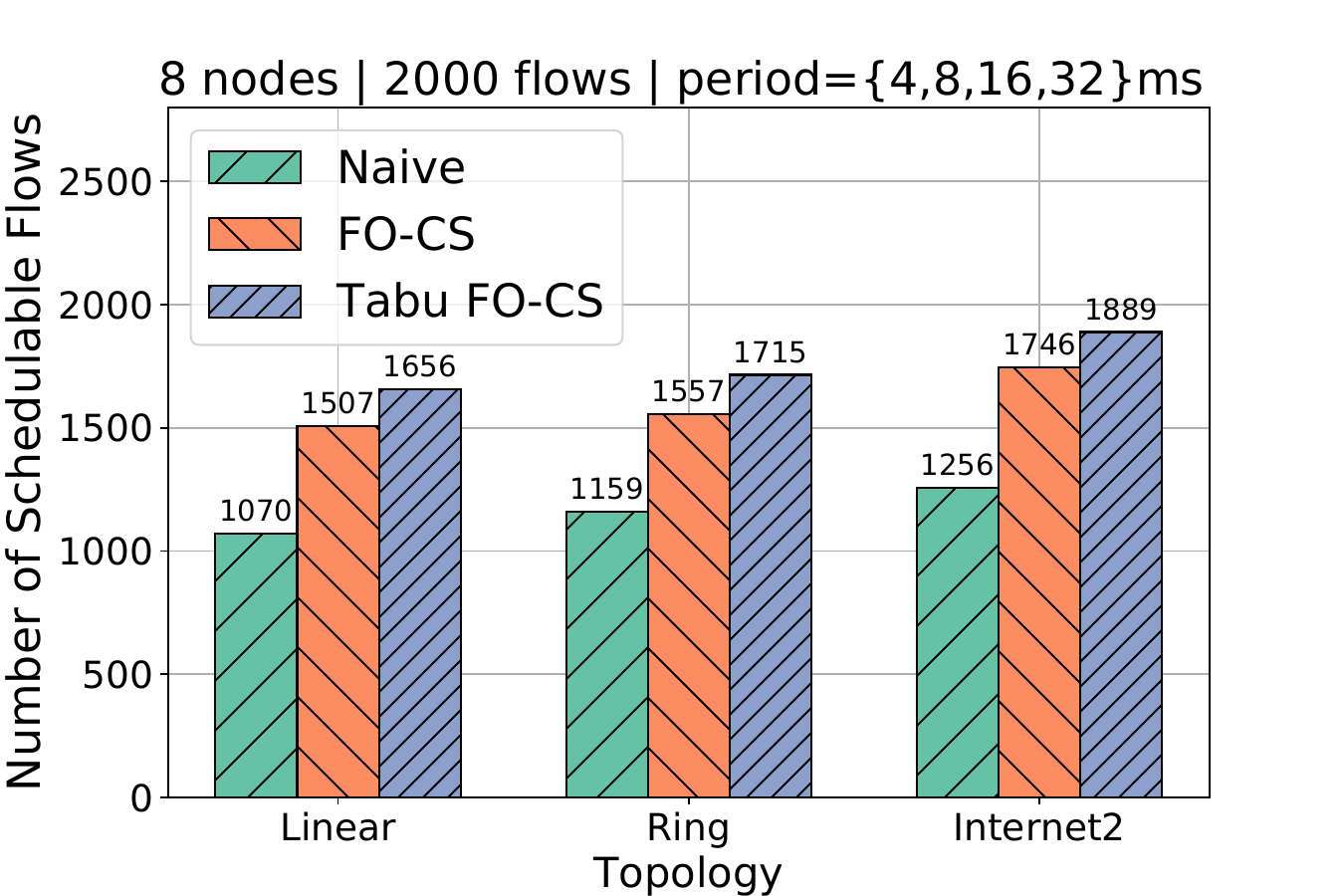}
			\label{fig12:evalutiona}
			% \caption{fig1}
		\end{minipage}
	}
	\subfigure[ Impact on the execution time.]{
		\begin{minipage}[t]{0.3\textwidth}
			\centering
			\includegraphics[width=\textwidth]{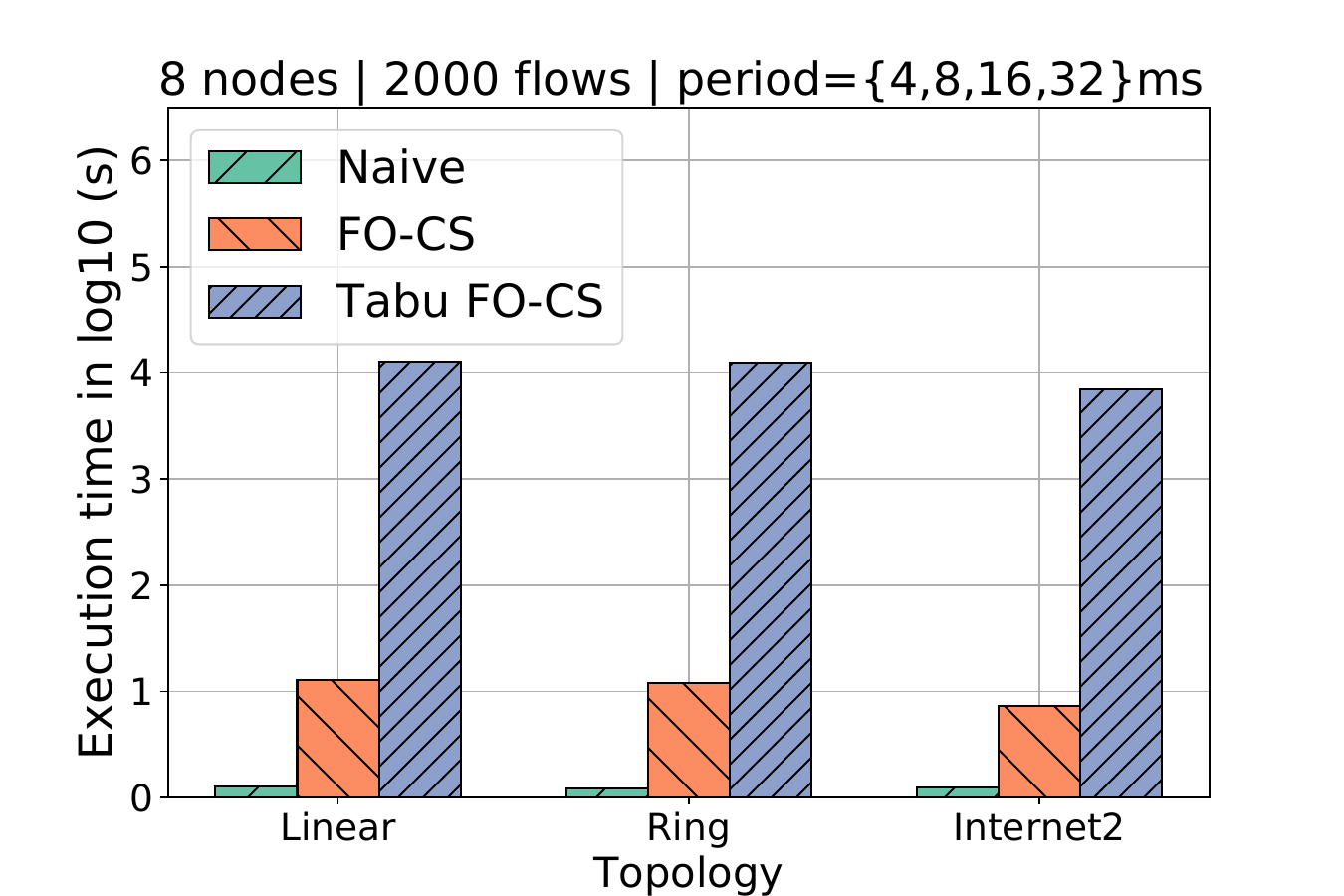}
			\label{fig12:evalutionb}
			% \caption{fig1}
		\end{minipage}
	}
	\subfigure[Effect of the link connectivity.]{
		\begin{minipage}[t]{0.3\textwidth}
			\centering
			\includegraphics[width=\textwidth]{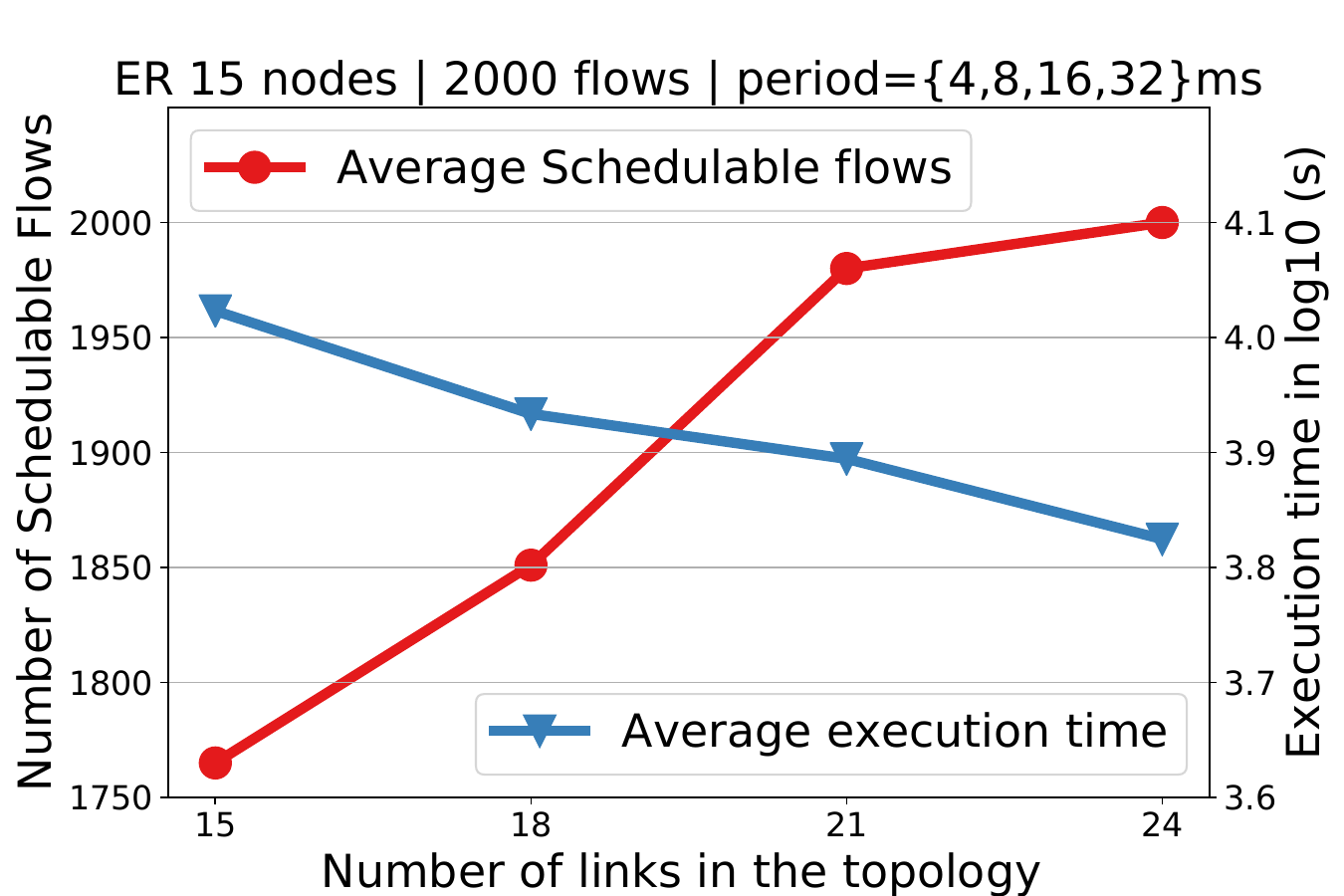}
			\label{fig12:evalutionc}
		\end{minipage}
	}
	\centering
	\caption{ Simulation results for the Tabu FO-CS algorithms under different topologies. 
		%(a) The number of scheduled flows and execution time per flow vs No. of links in the topology. (b) The number of average scheduled flows with two different periods set in small topology. (c) The average execution time and the maximum execution time with two different periods set in small topology. 
	}
\end{figure*}

\subsection{Simulation Results for Tabu FO-CS}
Furthermore, we evaluate the performance of Tabu FO-CS under different topologies, including the linear, the ring, and a part of Internet2. For the termination criteria in Algorithm 2, we set the maximum iterations $K$ to 1000 and the maximum repetitions $P$ to 100. The queue length of 10 and queue number of 4 are adopted for the FO-CS algorithm based on the previous evaluations. Each test is performed five times, and we calculate the average number of schedulable flows and the mean time consumed.

\subsubsection{ Impact on the Schedulability} 
In Fig. \ref{fig12:evalutiona}, we juxtapose the performance of the Tabu FO-CS algorithm against the FO-CS and naive algorithms. The evaluations show that the Tabu FO-CS algorithm can schedule 94.45\% of flows at the level of 2000 flows under the  topology of Internet2. In all topologies, the Tabu FO-CS algorithm consistently outperforms the FO-CS algorithm in scheduling outcomes. The maximal improvement is 7.9\% compared to FO-CS under the ring topology, and 31.65\% compared to the naive algorithm under the topology of Internet2.

\subsubsection{ Impact on the Execution Time} 
Fig. \ref{fig12:evalutionb}  illustrates the execution time of the above algorithms. We take the logarithmic function of log10 for the execution time as its value varies widely. The evaluations exhibit that the execution time of the Tabu FO-CS algorithm is on the order of ten thousand seconds, while the FO-CS algorithm is on the order of ten seconds. Thus, the Tabu FO-CS algorithm is suitable for offline optimization scenarios, and the FO-CS algorithm can be applied to the rapid generation of a scheduling solution.

\subsubsection{ Effect of the Link Connectivity} 
Another key performance of the algorithm is scalability. The number of links in the network will have a certain impact on the number of schedulable flows and algorithm execution time.  As depicted in Fig. \ref{fig12:evalutionc}, in the case of 15 links, the Tabu FO-CS algorithm successfully schedules 88.25\% of the flows with an average time cost of 2.9 hours. In the case of 24 links, the number of schedulable flows  improves by 11.75\% with a reduction of 36.68\% in the average execution time.  The reason is that  increased links provide more time slot resources and reduce the possibility of scheduling failures.

%To verify the scalability of the algorithm, we explore how the number of links in the ER topology affects the performance of the algorithm. As depicted in Fig. \ref{fig12:evalutionc}, when there are 15 links in the topology, the Tabu FO-CS algorithm could schedule 88.25\% flows with an average execution time of 2.9 hours. As the number of links increases, more resources of time slots are released to accommodate more flows, and the number of schedulable flows increases as a result. When there are 24 links in the topology, the schedulability improves by 11.75\% with a reduction of 36.68\% in the average execution time.

\section{Related Works}

There are plenty of related works about small-scale deterministic flow transmission in local Industrial IoT systems, but few works on that in large-scale Industrial Internet. Next, we mainly present related works around the CSQF-based time-sensitive scheduling mechanism.

%The deterministic transmission is highly desirable in industrial networks due to the time-critical requirements posed by industrial automation tasks.  There are a large number of related works about  deterministic flow scheduling in local-area industrial networks, but few works on that in wide-area networks. We summarize them and classify existing flow scheduling mechanisms into two classes based on the network scale. Since the CSQF is a recent mechanism proposed by the IETF DetNet work group, there is less work on it.  We only review some closely related ones here.

\textbf{Small-scale Time-sensitive Networks}: In automotive, aerospace, and Industrial IoT scenarios, the enhanced real-time Ethernet technologies, such as the Time-Triggered Ethernet and TSN,  are replacing the traditional fieldbus. Based on global time synchronization, TTE requires strict planning of the sending start time of each packet. The authors in \cite{steiner2010} studied the synthesis of scheduling time tables in TTE with SMT/OMT solver.  Furthermore,  TSN provides a series of time-based gate-controlled mechanisms to guarantee bounded-delay transmission. The majority of research on TSN centers around the creation and refinement of gate control lists\cite{explore_limits}\cite{PAFT}\cite{10121738}. Recently,  Yan \textit{et al}.\cite{itp} proposed planning methods to map the flows onto the ping-pong queue resources of CQF. 

%The authors in \cite{10196033}  further investigated the configuration of guard band and offsets for CQF, considering nonideal clocks and non-zero propagation times.

%\textbf{Small-scale Deterministic Networks}: Two main local-area deterministic transmission models based on Ethernet are TTE (Time-Triggered Ethernet) and TSN (Time-Sensitive Networking). Both TTE and TSN require precise time synchronization in switches and end devices. The schedule in the model of TTE is per-packet while the schedule in TSN is per-queue\cite{tteandtsn}.  Steiner \textit{et al}. focused on the schedule synthesis of time-triggered flows in TTE-enabled networks with SMT/OMT solver since 2010\cite{steiner2010}.  The IEEE TSN task group defines a set of standards for shaping and scheduling such as TAS (Time-Aware Shaper) and CQF (Cyclic Queuing and Forwarding). Most works on time-sensitive networks focus on the generation and optimization of gate control lists in TAS\cite{explore_limits} \cite{no_wait_packet_scheduling}. Recently,  Yan \textit{et al}. proposed an injection time planning mechanism that maps the time-sensitive flows onto the underlying resources both temporally and spatially to make CQF practical\cite{itp}.

\textbf{Long-distance Industrial IoT Networks}:  CQF is extended to the CSQF to enable the multi-queues scheduling by the IETF DetNet working group. The study \cite{Nasrallah_large_sim} examined the CQF and Paternoster mechanisms within a standard industrial control loop, taking into account different link delays. The authors of \cite{ctsdn} proposed a C-TSDN architecture, in which the CQF and CSQF protocols are jointly designed to realize seamless scheduling for factory IoT infrastructure. \cite{in_edge_control} managed to provide  layer-3 time-sensitive transmission down to the field level, which achieves remote industrial control in a heterogeneous network.  Moreover, Large-scale Deterministic Network (LDN)\cite{LDN} classifies cycle mapping into swap mode and stack mode, and emphasizes the ingress shaping function\cite{towards}. CSQF can be considered as the stack-mode-based LDN. 
%Asynchronous Traffic Shaper (ATS)\cite{9496182} is also increasingly recognized as one of the key techniques to eliminate the dependence on time synchronization. 
The authors of \cite{joint_large1} addressed the joint routing and scheduling issues for CSQF. They target to the path generation with feasible arc-cycle capacity, which cannot adaptively adjust the cycle tags. We solve the cycle tag computation problem, especially considering the long-distance link delay and adjacent cycle mapping relationship. To our knowledge, this is the pioneering study introducing the global cycle tag planning for the CSQF.

%In addition, there is a rough consensus that any DetNet solutions that require maintaining flow states at the intermediate nodes (such as IntServ) will not scale well with the increasing demands\cite{ppv}. Thus, the wide-area deterministic schemes must be core‑agnostic or core-stateless. Dampers\cite{damper}  are presented to reduce jitter by postponing packets for a specific time written in packet headers. LDN\cite{towards} scatters incoming bursts at ingress nodes and makes flows fit into the assigned cycles. PPV(Per Packet Value)\cite{ppv} provides guarantees for per-hop latency by encoding the utility function of flows to packet value markings. These mechanisms restrict the queuing delay by requiring packets to carry cycle or delay information, which are in the same vein as CSQF.

\section{Conclusion}

This study introduced the CTP,  a comprehensive integer programming framework tailored for CSQF-driven network-wide TS flow scheduling, which can make CSQF practical for long-distance Industrial IoT scenarios. The model decouples link delays from the cycle time, ensuring uniform queue rotation durations and addressing varying link delays by mapping the arriving cycle to a specific receiving queue. To avoid queue overflows under traffic aggregations, we devised the heuristic FO-CS algorithm for efficient cycle tag calculation.  Empirical evaluations underscored the efficacy of FO-CS, marking a notable 31.2\% enhancement in the number of schedulable flows compared to the Naive method.  The Tabu FO-CS algorithm approaches 94.45\% in scheduling 2000 flows.  In the future, we will optimize multiple objectives and leverage deep reinforcement learning to solve the model.

%\section*{Acknowledgment}

\bibliographystyle{IEEEtran}

\bibliography{IEEEabrv, reference}

\begin{thebibliography}{10}
\providecommand{\url}[1]{#1}
\csname url@samestyle\endcsname
\providecommand{\newblock}{\relax}
\providecommand{\bibinfo}[2]{#2}
\providecommand{\BIBentrySTDinterwordspacing}{\spaceskip=0pt\relax}
\providecommand{\BIBentryALTinterwordstretchfactor}{4}
\providecommand{\BIBentryALTinterwordspacing}{\spaceskip=\fontdimen2\font plus
\BIBentryALTinterwordstretchfactor\fontdimen3\font minus
  \fontdimen4\font\relax}
\providecommand{\BIBforeignlanguage}[2]{{%
\expandafter\ifx\csname l@#1\endcsname\relax
\typeout{** WARNING: IEEEtran.bst: No hyphenation pattern has been}%
\typeout{** loaded for the language `#1'. Using the pattern for}%
\typeout{** the default language instead.}%
\else
\language=\csname l@#1\endcsname
\fi
#2}}
\providecommand{\BIBdecl}{\relax}
\BIBdecl

\bibitem{iiot}
S.~Vitturi, C.~Zunino, and T.~Sauter, ``{Industrial communication systems and
  their future challenges: next-generation Ethernet, IIoT, and 5G},''
  \emph{Proceedings of the IEEE}, vol. 107, no.~6, pp. 944--961, 2019.

\bibitem{10293177}
M.~Ulbricht, S.~Senk, H.~K. Nazari, H.-H. Liu, M.~Reisslein, G.~T. Nguyen, and
  F.~H.~P. Fitzek, ``Tsn-flextest: Flexible tsn measurement testbed,''
  \emph{IEEE Transactions on Network and Service Management}, 2023.

\bibitem{remote_control}
I.~Pelle, F.~Paolucci, B.~Sonkoly, and F.~Cugini, ``Latency-sensitive
  edge/cloud serverless dynamic deployment over telemetry-based packet-optical
  network,'' \emph{IEEE Journal on Selected Areas in Communications}, vol.~39,
  no.~9, pp. 2849--2863, 2021.

\bibitem{smart_grid}
B.~Hu and H.~Gharavi, ``A hybrid wired/wireless deterministic network for smart
  grid,'' \emph{IEEE Wireless Communications}, vol.~28, no.~3, pp. 138--143,
  2021.

\bibitem{digital_twin}
M.~Groshev, C.~Guimarães, J.~Martín-Pérez, and A.~de~la Oliva, ``Toward
  intelligent cyber-physical systems: digital twin meets artificial
  intelligence,'' \emph{IEEE Communications Magazine}, vol.~59, no.~8, pp.
  14--20, 2021.

\bibitem{URLLC}
A.~Nasrallah, A.~S. Thyagaturu, Z.~Alharbi, C.~Wang, X.~Shao, M.~Reisslein, and
  H.~ElBakoury, ``{Ultra-low latency (ULL) networks: the IEEE TSN and IETF
  DetNet standards and related 5G ULL research},'' \emph{IEEE Communications
  Surveys Tutorials}, vol.~21, no.~1, pp. 88--145, 2019.

\bibitem{9788573}
L.~Zhao, P.~Pop, and S.~Steinhorst, ``Quantitative performance comparison of
  various traffic shapers in time-sensitive networking,'' \emph{IEEE
  Transactions on Network and Service Management}, vol.~19, no.~3, pp.
  2899--2928, 2022.

\bibitem{ITUR}
ITU-R, \emph{{IMT traffic estimates for the years 2020 to 2030}}, 2015,
  \url{https://www.itu.int/dms_pub/itu-r/opb/rep/R-REP-M.2370-2015-PDF-E.pdf}.

\bibitem{incre_tssdn}
N.~G. Nayak, F.~D{\"u}rr, and K.~Rothermel, ``{Incremental flow scheduling and
  routing in time-sensitive software-defined networks},'' \emph{IEEE
  Transactions on Industrial Informatics}, vol.~14, no.~5, pp. 2066--2075,
  2017.

\bibitem{load_balancing_csqf}
S.~{Chen}, J.~{Leguay}, S.~{Martin}, and P.~{Medagliani}, ``{Load balancing for
  deterministic networks},'' in \emph{Proc. 2020 IFIP Networking}, 2020, pp.
  785--790.

\bibitem{PCSQ}
Y.~Huang, S.~Wang, S.~Zhu \emph{et~al.}, ``Programmable cycle-specified queue
  for deterministic networking,'' in \emph{Proceedings of the ACM SIGCOMM
  Conference}, 2023, pp. 1132--1134.

\bibitem{ceni}
S.~Wang, B.~Wu, C.~Zhang, Y.~Huang, T.~Huang, and Y.~Liu, ``{Large-scale
  deterministic IP networks on CENI},'' in \emph{Proc. IEEE INFOCOM WKSHPS},
  2021, pp. 1--6.

\bibitem{itp}
J.~{Yan}, W.~{Quan}, X.~{Jiang}, and Z.~{Sun}, ``{Injection time planning:
  making CQF practical in time-sensitive networking},'' in \emph{Proc. IEEE
  INFOCOM}, 2020, pp. 616--625.

\bibitem{joint_large1}
J.~Krolikowski, S.~Martin, P.~Medagliani, J.~Leguay, S.~Chen, X.~Chang, and
  X.~Geng, ``Joint routing and scheduling for large-scale deterministic ip
  networks,'' \emph{Computer Communications}, vol. 165, pp. 33--42, 2021.

\bibitem{9854862}
Y.~Huang, S.~Wang, X.~Zhang, T.~Huang, and Y.~Liu, ``Flexible cyclic queuing
  and forwarding for time-sensitive software-defined networks,'' \emph{IEEE
  Transactions on Network and Service Management}, vol.~20, no.~1, pp.
  533--546, 2023.

\bibitem{explore_limits}
J.~Falk, F.~D{\"u}rr, and K.~Rothermel, ``{Exploring practical limitations of
  joint routing and scheduling for TSN with ILP},'' in \emph{Proc. IEEE
  Embedded and Real-Time Computing Systems and Applications}, 2018, pp.
  136--146.

\bibitem{in_edge_control}
A.~{Badar}, D.~Z. {Lou}, U.~{Graf}, C.~{Barth}, and C.~{Stich}, ``{Intelligent
  edge control with deterministic-IP based industrial communication in process
  automation},'' in \emph{Proc. Network and Service Management}, 2019, pp.
  1--7.

\bibitem{ts22261}
\emph{3GPP TS22261 v16.2.0. Service requirements for the 5G system}, 2017,
  \url{https://www.3gpp.org/ftp//Specs/archive/22_series/22.261/22261-g20.zip}.

\bibitem{use_case}
\BIBentryALTinterwordspacing
\emph{{IETF deterministic networking use cases}}. [Online]. Available:
  \url{https://datatracker.ietf.org/doc/rfc8578/.}
\BIBentrySTDinterwordspacing

\bibitem{ieee8021as}
``{IEEE} standard for local and metropolitan area networks - timing and
  synchronization for time-sensitive applications in bridged local area
  networks,'' \emph{IEEE Std 802.1AS-2011}, pp. 1--292, 2011.

\bibitem{synce}
J.-L. Ferrant, M.~Gilson, S.~Jobert, M.~Mayer, M.~Ouellette, L.~Montini,
  S.~Rodrigues, and S.~Ruffini, ``{Synchronous Ethernet}: a method to transport
  synchronization,'' \emph{IEEE Communications Magazine}, vol.~46, no.~9, pp.
  126--134, 2008.

\bibitem{towards}
B.~Liu, S.~Ren, C.~Wang \emph{et~al.}, ``{Towards large-scale deterministic IP
  networks},'' in \emph{IFIP Networking Conference}, 2021, pp. 1--9.

\bibitem{csqf_huang}
Y.~Huang, S.~Wang, T.~Feng, J.~Wang, T.~Huang, R.~Huo, and Y.~Liu, ``{Towards
  network-wide scheduling for cyclic traffic in IP-based deterministic
  networks},'' in \emph{2021 4th International Conference on Hot
  Information-Centric Networking (HotICN)}, 2021.

\bibitem{qbv_toc}
M.~Vlk, Z.~Hanzálek, K.~Brejchová, S.~Tang, S.~Bhattacharjee, and S.~Fu,
  ``Enhancing schedulability and throughput of time-triggered traffic in ieee
  802.1qbv time-sensitive networks,'' \emph{IEEE Transactions on
  Communications}, vol.~68, no.~11, pp. 7023--7038, 2020.

\bibitem{9212141}
A.~Larrañaga, M.~C. Lucas-Estañ, I.~Martinez, I.~Val, and J.~Gozalvez,
  ``Analysis of 5g-tsn integration to support industry 4.0,'' in \emph{IEEE
  ETFA}, 2020, pp. 1111--1114.

\bibitem{traffic_type_mapping}
\BIBentryALTinterwordspacing
{IEC 60802 Standard}. [Online]. Available:
  \url{http://grouper.ieee.org/groups/802/1/files/public/docs2019/60802-Hotta-Traffic-Types-Mapping-to-TSN-Mechanism-0119-v01.}
\BIBentrySTDinterwordspacing

\bibitem{model_tabu}
R.~Macchiaroli, S.~Mole, and S.~Riemma, ``Modelling and optimization of
  industrial manufacturing processes subject to no-wait constraints,''
  \emph{International Journal of Production Research}, vol.~37, no.~11, pp.
  2585--2607, 1999.

\bibitem{internet2}
J.~{Castillo-Velazquez}, D.~{Serrano-Martinez}, and A.~{Morales}, ``{Emulation
  of the connectivity of backbone and management for the layer 3 service of
  internet2: 2016 Topology},'' in \emph{Proc. IEEE Central America and Panama
  Convention}, 2017, pp. 1--4.

\bibitem{traffictype}
\BIBentryALTinterwordspacing
D.~Paul, A.~Astrit, and P.~David, ``Time sensitive networks for flexible
  manufacturing testbed - description of converged traffic types.'' [Online].
  Available: \url{https://hub.iiconsortium.org/tsn-converged-traffic-types.}
\BIBentrySTDinterwordspacing

\bibitem{9723445}
E.~Mohammadpour and J.-Y. Le~Boudec, ``Analysis of dampers in time-sensitive
  networks with non-ideal clocks,'' \emph{IEEE/ACM Transactions on Networking},
  vol.~30, no.~4, pp. 1780--1794, 2022.

\bibitem{steiner2010}
W.~Steiner, ``{An evaluation of SMT-based schedule synthesis for time-triggered
  multi-hop networks},'' in \emph{Proc. Real Time Syst. Symp.}, 2010, pp.
  375--384.

\bibitem{PAFT}
Z.~Feng, M.~Cai, and Q.~Deng, ``An efficient pro-active fault-tolerance
  scheduling of ieee 802.1qbv time-sensitive network,'' \emph{IEEE Internet of
  Things Journal}, vol.~9, no.~16, pp. 14\,501--14\,510, 2022.

\bibitem{10121738}
G.~Miranda, E.~Municio, J.~Haxhibeqiri, J.~Hoebeke, I.~Moerman, and J.~M.
  Marquez-Barja, ``Enabling time-sensitive network management over multi-domain
  wired/wi-fi networks,'' \emph{IEEE Transactions on Network and Service
  Management}, vol.~20, no.~3, pp. 2386--2399, 2023.

\bibitem{Nasrallah_large_sim}
A.~Nasrallah, V.~Balasubramanian, A.~S. Thyagaturu, M.~Reisslein, and
  H.~Elbakoury, ``{Large scale deterministic networking: a simulation
  evaluation},'' \emph{ArXiv}, vol. abs/1910.00162, 2019.

\bibitem{ctsdn}
Y.~Huang, S.~Wang, T.~Huang, and Y.~Liu, ``Cycle-based time-sensitive and
  deterministic networks: Architecture, challenges, and open issues,''
  \emph{IEEE Communications Magazine}, vol.~60, no.~6, pp. 81--87, 2022.

\bibitem{LDN}
Q.~li, X.~Geng, B.~Liu, T.~Eckert, L.~Geng, and G.~Li, \emph{{Large-Scale
  Deterministic IP Network}}, 2019,
  \url{https://datatracker.ietf.org/doc/html/draft-qiang-detnet-large-scale-detnet-05}.

\end{thebibliography}


\begin{thebibliography}{00}
\bibitem{b1} G. Eason, B. Noble, and I. N. Sneddon, ``On certain integrals of Lipschitz-Hankel type involving products of Bessel functions,'' Phil. Trans. Roy. Soc. London, vol. A247, pp. 529--551, April 1955.
\bibitem{b2} J. Clerk Maxwell, A Treatise on Electricity and Magnetism, 3rd ed., vol. 2. Oxford: Clarendon, 1892, pp.68--73.
\bibitem{b3} I. S. Jacobs and C. P. Bean, ``Fine particles, thin films and exchange anisotropy,'' in Magnetism, vol. III, G. T. Rado and H. Suhl, Eds. New York: Academic, 1963, pp. 271--350.
\bibitem{b4} K. Elissa, ``Title of paper if known,'' unpublished.
\bibitem{b5} R. Nicole, ``Title of paper with only first word capitalized,'' J. Name Stand. Abbrev., in press.
\bibitem{b6} Y. Yorozu, M. Hirano, K. Oka, and Y. Tagawa, ``Electron spectroscopy studies on magneto-optical media and plastic substrate interface,'' IEEE Transl. J. Magn. Japan, vol. 2, pp. 740--741, August 1987 [Digests 9th Annual Conf. Magnetics Japan, p. 301, 1982].
\bibitem{b7} M. Young, The Technical Writer's Handbook. Mill Valley, CA: University Science, 1989.
\end{thebibliography}

\end{document}